\documentclass[11pt,letterpaper]{article}
\usepackage[margin=1in]{geometry}
\usepackage{amsmath,amsthm,amssymb,amsfonts,amscd,amstext}
\usepackage{graphicx} 
\usepackage[table]{xcolor} 
\usepackage{mathtools}
\usepackage{url} 
\usepackage{xurl} 
\usepackage{authblk} 
\usepackage{algorithm} 
\usepackage{algpseudocode}
\usepackage{natbib} 
\usepackage{multirow}
\usepackage[font=scriptsize]{caption}
\usepackage[font=scriptsize]{subcaption}
\usepackage{enumitem}
\usepackage[htt]{hyphenat} 
\usepackage{longtable}
\usepackage{thmtools,thm-restate}
\usepackage{setspace}
\usepackage{textcomp}
\usepackage[colorinlistoftodos]{todonotes}
\usepackage{booktabs}
\usepackage{tabularx}
\usepackage{array}

\newcounter{savealgorithm}

\DeclareMathOperator*{\argmax}{arg\,max}

\makeatletter
\newcommand{\algmargin}{\the\ALG@thistlm}
\makeatother

\algrenewcommand\algorithmicrequire{\textbf{Input:}}
\algrenewcommand\algorithmicensure{\textbf{Output:}}

\algnewcommand{\parState}[1]{\State%
    \parbox[t]{\dimexpr\linewidth-\algmargin}{\strut\hangindent=\algorithmicindent \hangafter=1 #1\strut}}

\theoremstyle{plain}
\newtheorem{Theorem}{Theorem}
\newtheorem{Assumption}{Assumption}
\newtheorem{Condition}{Condition}

\newtheorem{Proposition}{Proposition}
\newtheorem{Lemma}{Lemma}

\theoremstyle{remark}

\newcommand{\sign}{\operatorname{sign}}

\allowdisplaybreaks
\begin{document}

\title{Response-guided knockoffs for prespecified directional FDR control in linear models}

\author[1]{Jack Freestone\thanks{Address for correspondence: Jack Freestone, School of Mathematical and Physical Sciences, Macquarie University, New South Wales 2109, Australia. Email: \texttt{jack.freestone@mq.edu.au}}}
\author[2]{Garth Tarr}
\author[1,2]{Samuel Muller}
\author[2]{Uri Keich}
\affil[1]{School of Mathematical and Physical Sciences, Macquarie University, New South Wales 2109, Australia}
\affil[2]{School of Mathematics and Statistics F07, University of Sydney, New South Wales 2006, Australia}

\date{}
\begingroup
\setstretch{1}
\maketitle
\endgroup
\thispagestyle{empty}

\begin{abstract}
Conventional false discovery rate (FDR) control limits the expected proportion of selected features that are unrelated to the response.
Directional FDR extends this criterion to also account for incorrect sign estimates.
In some applications, however, discoveries are sought only in a direction specified in advance.
We therefore introduce the \emph{prespecified} directional FDR, which counts as errors both null effects and effects opposite to the target direction in context of the fixed-design Gaussian linear model with unknown noise variance.
A straightforward adaptation of the knockoff filter provides prespecified directional FDR control.
To improve power, we develop the \emph{response-guided knockoff filter}, which uses a noise-perturbed response to guide knockoff construction in favour of features likely to have the target sign, while retaining finite-sample prespecified directional FDR control.
Notably, our improved filter maintains this guarantee with a sample size $n>p+2$ in contrast to standard constructions that require $n\geq2p$.
Simulations across a range of settings demonstrate power gains of the response-guided knockoff filter.
An application to HIV drug-resistance data illustrates improved recovery of externally annotated resistance-associated mutations.
\end{abstract}


\clearpage
\pagenumbering{arabic}
\section{Introduction}
\label{section:introduction}
In many scientific applications, the goal of feature selection is inherently one-sided.
For example, when studying genetic mutations associated with HIV drug resistance~\citep{rhee:genotypic}, investigators may seek mutations that \emph{increase} resistance to a particular drug while mutations that decrease resistance, or have no effect, are not of primary interest.
Similarly, in genomics one may wish to identify genes that are \emph{upregulated} under a specific treatment~\citep{liang2024directional}.
The goal in such settings is not simply to identify which features are relevant to a response, but to select those whose effects act in a prespecified direction $d \in \{\pm 1\}$, while controlling the rate of false discoveries among these directional claims.

We work in the classical linear model, where we observe $n$ independent data points from
\begin{equation}
    \label{eq:linear_model}
    Y = \sum_{j = 1}^p X_j \beta_j + \varepsilon,
\end{equation}
with unknown coefficients $\boldsymbol{\beta} \equiv (\beta_1, \dots, \beta_p)^\top \in \mathbb{R}^p$ and $\varepsilon \sim \mathcal N(0, \sigma^2)$.
A feature $j$ is \emph{relevant} when $\beta_j \neq 0$, i.e., $j \in \mathcal H_1 \equiv \{j \in [p] : \beta_j \neq 0\}$, where $[p] \equiv \{1, \ldots, p\}$.
A standard goal is to control the \emph{false discovery rate} (FDR)~\citep{benjamini:controlling} at a user-specified level $\alpha \in \left(0, 1\right)$:
\begin{align}
    \label{eq:fdr}
    \text{FDR} \equiv \mathbb{E}\left( \text{FDP} \right) \leq \alpha, \quad \text{FDP} \equiv \frac{ \lvert \mathcal H_0 \cap \mathcal {\hat S} \rvert }{\lvert \mathcal {\hat S} \rvert \vee 1},
\end{align}
where $\mathcal {\hat S} \subseteq [p]$ is the set of selected features, $\mathcal H_0 \equiv \{ j \in [p] : \beta_j = 0 \}$ is the set of irrelevant features, and $x \vee y \equiv \max\{ x, y\}$.

The knockoff filter of \cite{barber:controlling} controls this FDR in finite samples by constructing artificial negative controls, called \emph{knockoffs}, that mimic the correlation structure of the original features. 
In their follow up work, \cite{barber2019knockoff} further showed that the knockoff filter controls a stronger error metric called \emph{directional FDR}, which penalises both irrelevant discoveries and incorrect sign estimates:
\begin{align}
   \label{eq:fdr_sign}
    \text{FDR}_{\text{dir}} \equiv \mathbb{E}\left( \text{FDP}_{\text{dir}} \right) \leq \alpha, \quad \text{FDP}_{\text{dir}} \equiv \frac{ \lvert \{ j \in \mathcal {\hat S} : \widehat{\sign(\beta_{j})} \neq \sign\left( \beta_j \right) \} \rvert }{\lvert \mathcal {\hat S} \rvert \vee 1},
\end{align}
where $\widehat{\sign(\beta_{j})}$ is an estimate of the sign of $\beta_j$.
However, this guarantee treats positive and negative effects
symmetrically---it does not allow the analyst to target a specific direction of interest.

For a prespecified target direction $d \in \{\pm 1\}$ and a selected set $\mathcal{\hat S}\subseteq[p]$, we count a selected feature as a false discovery whenever $\sign(\beta_j)\neq d$.
We denote the corresponding error metric by $\text{FDR}_{\text{dir}}^{d}$:
\begin{align}
    \label{eq:fdr_dir_d}
    \text{FDR}_{\text{dir}}^{d} &\equiv \mathbb{E}\left(\text{FDP}_{\text{dir}}^{d}\right), \quad
    \text{FDP}_{\text{dir}}^{d} \equiv \frac{\lvert \mathcal H_0^d \cap \mathcal{\hat S} \rvert}{\lvert \mathcal{\hat S} \rvert \vee 1}, \\
    \mathcal H_0^d &\equiv \{ j \in [p] : \sign(\beta_j) \neq d \}. \nonumber
\end{align}
This is simply the usual FDR with respect to the null set $\mathcal H_0^d$, which includes both zero effects and effects in the opposite direction.
We use the term \emph{prespecified directional FDR} and the superscript $d$ to make the prespecified target direction explicit.
For a procedure that selects only features whose estimated signs equal $d$, $\text{FDR}_{\text{dir}}^{d}$ coincides with the directional FDR in~\eqref{eq:fdr_sign} evaluated on the same selected set $\mathcal{\hat S}$.

We first develop a \emph{signed-knockoff filter} that selects features with provable control of~\eqref{eq:fdr_dir_d}.
Like existing filters, the knockoff construction is limited by a budget shared across all its features.
Part of this budget is therefore wasted on features with no effect, or an effect in the wrong direction, which then compete with relevant features for selection at the inference step.
Fortunately, only a fraction of the information from the response is often enough to screen for features whose effects are in the target direction.
We can therefore sacrifice a small part of the response information to concentrate the budget on these features, reserving the rest of the response for inference.

To this end, we introduce the \textit{response-guided knockoff filter}, which screens features using a noise-perturbed copy of the response before constructing the knockoffs.
To retain prespecified directional FDR control of \eqref{eq:fdr_dir_d} after this peek at the response, we impose additional constraints so that each feature-knockoff pair remain indistinguishable with respect to the information already used, giving the original features no advantage at the inference step.
Because screened-out features need no knockoffs of their own, the procedure achieves finite-sample prespecified directional FDR control whenever $n > p + 2$, whereas existing fixed-X knockoff constructions require $n \geq 2p$.
In simulations and an HIV drug-resistance study, it achieves higher power in the target direction than the signed-knockoff filter.

\subsection{Review of Fixed-X Knockoff filter}
\label{section:overview_FX}
Suppose we have $n$ independent observations $\{(X_{i1}, \dots, X_{ip}, Y_i): i \in [n] \}$ from the linear model in (\ref{eq:linear_model}).
Denote $\mathbf y \equiv (Y_1, \dots, Y_n)^\top \in \mathbb R^n$ as the response vector and $\mathbf X \equiv (X_{ij}) \in \mathbb{R}^{n \times p}$ as the fixed design matrix.
Throughout we assume that the columns of $\mathbf X$, denoted as $\mathbf X_j$, are normalised so that $\|\mathbf X_j\|_2 = 1$
for all $j \in [p]$.

The \textit{fixed-X} knockoff filter begins by constructing a knockoff design matrix $\tilde{\mathbf X}$
that satisfies the key requirement that each column $\mathbf X_j$, and its corresponding knockoff, denoted as $\tilde{\mathbf X}_j$,
share the same correlation to all other original features and knockoffs:
\begin{align}\label{eq:corr_structure}
    \mathbf X_i^{\!\top} \mathbf X_j = \tilde{\mathbf X}_i^{\!\top} \mathbf X_j = \tilde{\mathbf X}_i^{\!\top} \tilde{\mathbf X}_j \quad \forall i \neq j \quad i, j \in \left[ p \right].
\end{align}
This construction guarantees that for $j$ such that $\beta_j = 0$, $\mathbf X_j$ and $\tilde{\mathbf X}_j$ will, in probability,
do equally well at explaining the variance in $\mathbf y$, since both features share the same correlations with all other columns.

One can trivially construct such knockoffs by taking $\tilde{\mathbf X} = \mathbf X$. However, this would be counterproductive
because the procedure is predicated on the premise that for $j$ with $\beta_j \neq 0$, $\mathbf X_j$ directly contributes
to $\mathbf y$, and therefore we expect this column to do better at explaining the variance in $\mathbf y$ than $\tilde{\mathbf X}_j$.
How much better is determined by how distinguishable $\tilde{\mathbf X}_j$ is from $\mathbf X_j$, which is governed by the
\textit{separation parameter} $s_j\in[0,1]$ in the inner product
\begin{align}\label{eq:corr_away}
    \mathbf X_j^{\!\top} \tilde{\mathbf X}_j = 1 - s_j \geq 0 \quad j \in \left[ p \right].
\end{align}
\cite{barber:controlling} showed that we can construct knockoffs that satisfy \eqref{eq:corr_structure} and \eqref{eq:corr_away} if and only if
the following condition holds.
\begin{Condition}[Feasibility]
\label{ass:gram}
There exist separation parameters $\mathbf s = (s_j)_{j=1}^p \in [0,1]^p$ such that with $\boldsymbol{\Sigma} \equiv {\mathbf X}^\top \mathbf X$
\begin{equation}
	\label{def:G}
	\mathbf G  \equiv \begin{bmatrix}
        \boldsymbol{\Sigma} & \boldsymbol{\Sigma} - \operatorname{diag}\{\mathbf s\} \\
        \boldsymbol{\Sigma} - \operatorname{diag}\{\mathbf s\} & \boldsymbol{\Sigma} 
\end{bmatrix} \succeq 0 .
\end{equation}
\end{Condition}
In practice, we typically first determine $\mathbf s\in [0,1]^p$ so that $\mathbf G\succeq0$ (which is equivalent to
$2 \boldsymbol{\Sigma} - \operatorname{diag}\{\mathbf s\} \succeq 0$), and then construct $\tilde{\mathbf X}$ so that with
$\hat{\mathbf X} \equiv [\mathbf X, \tilde{\mathbf X}]$, $\mathbf G  = \hat{\mathbf X}^\top \hat{\mathbf X}$ \citep{barber:controlling}.

The separation parameters are chosen by solving one of several constrained optimisation problems,
each trying to maximise the separation in some sense, while ensuring that Condition~\ref{ass:gram} is
satisfied~\citep{barber:controlling,gimenez:improving,spector2022powerful}.
The simplest example is the equi-correlated knockoff-construction, which sets $s_j \equiv s$ for all $j$ and maximises $s$ subject to Condition~\ref{ass:gram}.


Once the knockoffs have been constructed, the features are assigned statistics $\mathbf W \equiv \mathbf W([\mathbf X, \tilde{\mathbf X} ], \mathbf y) \in \mathbb{R}^p$ where large positive statistics are used as evidence that the corresponding feature is relevant to the model. 
$\mathbf W$ needs to satisfy the following conditions.
\begin{Condition}[Sufficiency]
    \label{ass:suff}
$\mathbf W$ depends on $\mathbf X, \tilde{\mathbf X}, \mathbf y$ only through $\hat{\mathbf X}^{\!\top}\mathbf y$ and $\hat{\mathbf X}^{\!\top}\hat{\mathbf X}$.
\end{Condition}

\begin{Condition}[Anti-symmetry]
    \label{ass:sym}
Swapping a subset of the original features for their knockoffs only changes the sign of the score:
	\begin{equation}\label{eq:anti-symmetry}
		W_{\Pi(j)}(\hat{\mathbf X} \circ \Pi, \mathbf y) = W_j(\hat{\mathbf X}, \mathbf y) \cdot \begin{cases}
            +1  \quad \Pi(j) = j\\
            -1 \quad \Pi(j) = j + p \mod 2p
         \end{cases} \\
	\end{equation}
where $\Pi$ is a permutation on the columns of the extended design matrix $\hat{\mathbf X}$ that swaps a subset of the original features for their corresponding knockoffs, i.e., $\forall j$, $\Pi(j) \in \{j, j + p \mod 2p\}$.
\end{Condition}


Under the linear model assumption, and using Conditions~\ref{ass:gram}--\ref{ass:sym}, it can be shown that the scores satisfy the following distribution assumption with $\mathcal N_0 = \mathcal H_0$ and exact equality at $c = 1/2$.
\begin{Assumption}
    \label{assumption:general}
    Let $\mathcal N_0 \subseteq [m]$ index the true null hypotheses.
    Conditioned on $\lvert \mathbf W \rvert$ and the signs of $\{W_i : i \notin \mathcal N_0\}$, the non-vanishing signs of $\{W_i : i \in \mathcal N_0\}$ are independent $\pm 1$ random variables with $\mathbb{P}(W_i > 0) \leq c$.
\end{Assumption}

Hence, using this result, the number of features that score lower than $-t<0$ can be used to conservatively estimate the number of irrelevant features that score higher than $t>0$, and therefore estimate an upper bound for the false discovery proportion as follows
\begin{align}
    \label{eq:fdr_estimate}
    \widehat{\text{FDR}}(t) \equiv \frac{\#\{W_j \leq -t \} + 1}{ \#\{W_j \geq t \} \vee 1 } &> \frac{\#\{W_j \leq -t : j \in \mathcal H_0 \}}{ \#\{W_j \geq t \} \vee 1 }\\
    & \approx  \frac{\#\{W_j \geq t : j \in \mathcal H_0 \}}{ \#\{W_j \geq t \} \vee 1 } = \text{FDP}(t),
\end{align}
where the ``$+1$'' included in the numerator of $\widehat{\text{FDR}}(t)$ is sufficient for theoretical finite-sample FDR control~\citep{barber:controlling} (and in general, necessary \citep{rajchert:controlling}).
Given a prespecified threshold $\alpha$, the knockoff filter reports the following subset of features as relevant:
\begin{align}
    \label{discovery}
    \mathcal{\hat{S}}_{\text{KF}} \equiv \mathcal{\hat{S}}_{\text{KF}}(\alpha) \equiv \{j \in [p] : W_j \geq \tau(\alpha) \}, \quad \tau \equiv \tau (\alpha) \equiv \inf \{t > 0 : \widehat{\text{FDR}}(t) \leq \alpha \}.
\end{align}
This last step, where the statistics are converted to a discovery list, i.e., $\mathbf W \mapsto \mathcal{\hat{S}}_{\text{KF}}$, is the procedure called \emph{Selective SeqStep+}, abbreviated as SSS+ (Algorithm~\ref{alg:SeqStep}).
\cite{barber:controlling} presented SSS+ as a more general procedure for multiple hypothesis testing and later extended the FDR guarantees to null signs that are independent but not identically distributed, so that SSS+ controls the FDR whenever the statistics satisfy Assumption~\ref{assumption:general} \citep{barber2019knockoff}.
In particular, the knockoff filter applies SSS+ with $c = 1/2$, for which Assumption~\ref{assumption:general} holds with equality for $\mathcal N_0 = \mathcal H_0$, and therefore controls the FDR at level $\alpha$.

The remainder of this paper is organised as follows.
Section~\ref{section:methods} shows why naive adaptations of the knockoff filter fail to control the prespecified directional FDR when targeting discoveries in a specific direction, and then develops two principled methods: (i) the \emph{signed-knockoff filter}, a modification of the knockoff filter that selects features in a target direction with prespecified directional FDR control, and (ii) the \emph{response-guided knockoff filter}, which is designed to increase power by incorporating response information into the knockoff-construction. 
Section~\ref{section:results} reports simulation and real-data experiments.
Finally, Section~\ref{section:discussion} interprets the findings, discusses limitations of our methods, and highlights promising extensions for further study.

\begin{algorithm}[H]
	\caption{ {\bf Selective SeqStep+ (SSS+)} \citep{barber:controlling}}
	\label{alg:SeqStep}
	\begin{algorithmic}[1]
		\Require  \begin{tabular}[t]{p{0.8\textwidth}}
		$\mathbf W \equiv (W_i)_{i = 1}^m$ the list of statistics with ties broken randomly; $c \in (0, 1)$ such that $\mathbf W$ satisfies Assumption~\ref{assumption:general}; $\alpha \in (0, 1)$ the FDR threshold
		\end{tabular}
		\Ensure A discovery list $\mathcal{\hat S}_{\text{KF}}$
			\State $A_t \gets \#\{i \in [m] : W_i \leq -t \}$ 
			\State $R_t \gets \#\{i \in [m] : W_i \geq t \}$ 
    		\State $\tau \gets \inf \{t > 0:  \frac{A_t + 1} {R_t \vee 1} \cdot \frac{c}{1 - c} \leq \alpha \}$	
		\State \Return $\mathcal{\hat S}_{\text{KF}} \gets \{j \in [p] : W_j \geq \tau \}$
		\end{algorithmic}
\end{algorithm}

\section{Methods}
\label{section:methods}

\subsection{Naive approaches to prespecified directional FDR control}
\label{section:naive-approaches}
\citet{barber2019knockoff} showed that the knockoff filter controls the directional FDR~\eqref{eq:fdr_sign} when each discovery~$j$ is assigned its own sign estimate
\begin{equation}
	\label{eq:bc_sign_est}
\widehat{\sign(\beta_j)} \equiv \sign\bigl((\mathbf X_j - \widetilde{\mathbf X}_j)^{\!\top}\mathbf y\bigr).
\end{equation}
That result, however, does not address the setting where only discoveries with a prespecified sign $d\in\{\pm 1\}$ are of interest.
We show in Figure~\ref{fig:naive_approaches} that the unmodified knockoff filter substantially violates the prespecified directional FDR control in~\eqref{eq:fdr_dir_d}.
This is expected: its scores satisfy Assumption~\ref{assumption:general} with $\mathcal N_0 = \mathcal H_0$, but not with $\mathcal N_0 = \mathcal H_0^d$.
To further motivate our new procedures, we first show that two natural adaptations of standard knockoff-based strategies that do take the estimated
sign into account still fail to control the prespecified directional FDR.

\subsubsection{Post-hoc sign screening}
A natural refinement of \citet{barber2019knockoff}'s directional-FDR approach is to start with the same list of discoveries $\mathcal{\hat S}_{\text{KF}}$ that the standard knockoff filter produces and then discard every discovery whose estimated sign as defined in \eqref{eq:bc_sign_est} differs from $d$, i.e., report
\[
\mathcal{\hat S}_{\text{post}} \equiv \bigl\{j \in \mathcal{\hat S}_{\text{KF}} : \widehat{\sign(\beta_j)} = \sign\bigl((\mathbf X_j - \widetilde{\mathbf X}_j)^{\!\top}\mathbf y\bigr) = d\bigr\}.
\]
While this removes some wrong-sign discoveries, the threshold $\tau$ in (\ref{discovery}) is computed using the original, unscreened statistics.
The post-hoc removal reduces the denominator $\lvert \mathcal{\hat S} \rvert$ without a corresponding correction to $\widehat{\text{FDR}}(t)$, so the prespecified directional FDR can still be inflated (see Figure~\ref{fig:naive_approaches}).

\subsubsection{Pre-screening by estimated sign}
Another strategy restricts the input to SSS+ to only those features with the desired estimated sign.
Specifically, define $\mathcal J_{\text{pre}} \equiv \{j \in [p] : \widehat{\sign(\beta_j)} = d \}$
and apply SSS+ to the screened scores $(W_j)_{j \in \mathcal J_{\text{pre}}}$ with $c = 1/2$.
The problem with this approach is that $\widehat{\sign(\beta_j)}$ and $\sign(W_j)$ are positively associated: a feature with $\widehat{\sign(\beta_j)} = +1$ implies $\mathbf X_j^\top \mathbf y > \widetilde{\mathbf X}_j^\top \mathbf y$, and $W_j$ is more likely to be positive even for true null features.
Therefore, restricting to the set $\mathcal J_{\text{pre}}$ with $d = +1$ skews the distribution of $W_j$ positively, so Assumption~\ref{assumption:general} with the screened null hypotheses $\mathcal H_0^d \cap \mathcal J_{\text{pre}}$ is violated.
Figure~\ref{fig:naive_approaches} confirms prespecified directional FDR control can be compromised with this approach.

\subsection{Signed-knockoff filter for prespecified directional FDR control}
\label{section:signed-knockoff-filter}

Our first procedure for controlling the prespecified directional FDR for target sign $d$, the \textit{signed-knockoff filter}, consists of pre-processing the scores according to a sign estimate different from~\eqref{eq:bc_sign_est} and running the knockoff filter.
Using the sign estimate $\widehat{\sign(\beta_{j})}$ from \eqref{eq:bc_sign_est}, we define
\begin{equation}
\label{eq:sign}
S_j = \widehat{\sign(\beta_{j})} \cdot \sign(W_j) = \sign\bigl((\mathbf X_{j}-\widetilde{\mathbf X}_{j})^{\!\top}\mathbf y\bigr) \cdot \sign\left(W_j \right).
\end{equation}
The quantity $S_j$ was originally introduced by \cite{barber2019knockoff} as a device in proving directional FDR control in \eqref{eq:fdr_sign} of the standard knockoff filter. Here we repurpose it as a criterion for filtering discoveries by their estimated effect direction and apply SSS+ to the list of features with $\mathcal J \equiv \{ j \in [p] : S_j = d \}$ (Algorithm~\ref{alg:signed-kf}).

\begin{algorithm}[H]
	\caption{ {\bf Signed-knockoff filter}}
	\label{alg:signed-kf}
	\begin{algorithmic}[1]
		\Require  \begin{tabular}[t]{p{0.8\textwidth}}
		$\mathbf y \in \mathbb R^n$ the response vector; $\mathbf X \in \mathbb R^{n \times p}$ the design matrix with $n \geq 2p$;  $d \in \{\pm 1\}$ the prespecified target sign of interest; $\alpha \in (0, 1)$ the prespecified directional FDR threshold\\
		\end{tabular}
		\Ensure A discovery list $\mathcal{\hat S}_{\text{signed-KF}}$
        \State $\mathbf s \gets s(\mathbf X) \in [0,1]^p$ \Comment{$\mathbf s$ satisfies Condition~\ref{ass:gram}}
        \State $\tilde{\mathbf X} \gets g(\mathbf X, \mathbf s)$ \Comment{$\tilde{\mathbf X}$
			satisfies $\mathbf G  = {[\mathbf X, \tilde{\mathbf X}]}^\top [\mathbf X, \tilde{\mathbf X}]$, $g$ a knockoff generator}
        \State $\hat{\mathbf X} \gets [\mathbf X,  \tilde{\mathbf X}]$
        \State $\mathbf W \gets \mathbf W(\hat{\mathbf X}, \mathbf y)$ \Comment{$\mathbf W$ satisfies Conditions~\ref{ass:suff} and \ref{ass:sym}}
        \State $S_j \gets \sign\bigl((\mathbf X_{j}-\widetilde{\mathbf X}_{j})^{\!\top}\mathbf y\bigr) \cdot \sign\left(W_j \right)$ for $j = 1, \ldots, p$
        \State $\mathcal J \gets \{ j \in [p] : S_j = d \}$
        \State \Return $\mathcal{\hat S}_{\text{signed-KF}} \gets \text{SSS+}((W_j)_{j \in \mathcal J}, 1/2, \alpha)$ \Comment{Algorithm~\ref{alg:SeqStep}}
		\end{algorithmic}
\end{algorithm}

We give the following intuition for why this procedure controls the prespecified directional FDR.
By rearranging \eqref{eq:sign}, any feature $j \in \mathcal J$ satisfies $\sign(W_j) = d \cdot \sign\bigl((\mathbf X_{j}-\widetilde{\mathbf X}_{j})^{\!\top}\mathbf y\bigr)$.
Hence, $W_j > 0$ occurs when $\sign\bigl((\mathbf X_{j}-\widetilde{\mathbf X}_{j})^{\!\top}\mathbf y\bigr)$ equals the target sign $d$.
It can be shown that $(\mathbf X_{j}-\widetilde{\mathbf X}_{j})^{\!\top}\mathbf y$ is normally distributed with mean $s_j \beta_j$, and since $s_j \geq 0$, its sign equals $d$ with probability at most $1/2$ whenever $\operatorname{sign}(\beta_j) \neq d$, i.e., for $j \in \mathcal H_0^d$.
Using results of \citet{barber2019knockoff}, we show that this bound continues to hold conditionally on information that fixes $\mathcal J$, so that the inputs $(W_j)_{j \in \mathcal J}$ satisfy Assumption~\ref{assumption:general} with $\mathcal N_0 = \mathcal H_0^d \cap \mathcal J$ and $c = 1/2$.
Hence SSS+ controls the prespecified directional FDR, as stated in Proposition~\ref{lem:signed-kf} below where the proof is deferred to Appendix~\ref{appendix:proofs}.


\begin{Proposition}
    \label{lem:signed-kf}
    Assume $\mathbf y \sim \mathcal{N}(\mathbf X\boldsymbol \beta, \sigma^2 \mathbf I_n)$ and $n \geq 2p$.
    Then the signed-knockoff filter $\mathcal{\hat S}_{\text{signed-KF}}$ controls the prespecified directional FDR defined in~\eqref{eq:fdr_dir_d} at level $\alpha$.
\end{Proposition}

\subsection{Response-guided knockoff filter}
\label{section:response-guided_kf}

Existing knockoffs are constructed solely from the design matrix $\mathbf X$, without incorporating any information from the response $\mathbf y$.
Instead, we propose to use part of the information in $\mathbf y$ to guide the construction of the knockoffs, so we can focus on the more promising features. 
Our method, which we call the \emph{response-guided knockoff filter} (Algorithm~\ref{alg:response-kf}), begins by adding noise to the response $\mathbf y$ to create a perturbed response, denoted as $\tilde{\mathbf y}$.
Specifically, $\tilde{\mathbf y} \equiv \mathbf y + \tilde{\boldsymbol{\varepsilon}}$, where given $\mathbf y$, $\tilde{\boldsymbol{\varepsilon}} \sim \mathcal N(\mathbf 0, s_y^2 \mathbf I_n)$ and $s_y^2$ is the sample variance of $\mathbf y$.
Using the perturbed response $\tilde{\mathbf y}$ and $\mathbf X$ we can estimate which features are both relevant to the response and have the appropriate sign.
The separation parameters $\mathbf s$ are subsequently chosen as a function of
$(\mathbf X, \tilde{\mathbf y}, d)$ (Step~2 of Algorithm~\ref{alg:response-kf}), so that features with an estimated sign $d$ receive larger separation from their knockoffs, at the expense of the remaining features.
The algorithm places no restriction on how this is done beyond Condition~\ref{ass:correct} below. 
In the implementation we use (see Section~\ref{sec:choose-s}), features judged to be inconsistent with the target sign are assigned $s_j = 0$, so that $\tilde{\mathbf X}_j$ coincides with $\mathbf X_j$ and the feature is effectively screened out, leaving room for larger $s_j$ on the rest. 

To account for peeking at the response through $\tilde{\mathbf y}$, our knockoff-construction must satisfy the following constraint, which extends Condition~\ref{ass:gram}.
\begin{Condition}[Perturbed response feasibility]
    \label{ass:correct}
    Let $\mathbf X^{(e)} \equiv [\,\mathbf X \ \frac{1}{\|\tilde{\mathbf y}\|_2}\tilde{\mathbf y} \ \frac{1}{\sqrt{n}}\mathbf 1 \,]$ denote the ``augmented design matrix'', where $\mathbf 1$ denotes a vector of ones, and let
	$\boldsymbol{\Sigma}^{(e)} \equiv {\mathbf X^{(e)}}^{\!\top}\mathbf X^{(e)}$. 
    There exist separation parameters $\mathbf s\in [0,1]^p$ such that 
    \[  \mathbf G^{(e)}  \equiv \begin{bmatrix}
        \boldsymbol{\Sigma}^{(e)} & \boldsymbol{\Sigma}^{(e)} - \operatorname{diag}\{\mathbf s, 0, 0\} \\
        \boldsymbol{\Sigma}^{(e)} - \operatorname{diag}\{\mathbf s, 0, 0\} & \boldsymbol{\Sigma}^{(e)}
    \end{bmatrix} \succeq 0. \]
\end{Condition}
\noindent As in the case of Condition~\ref{ass:gram}, $\mathbf G^{(e)}\succeq 0$ if and only if
$2\boldsymbol{\Sigma}^{(e)} - \operatorname{diag}\{\mathbf s, 0, 0\} \succeq 0$.
Similarly, if $\mathbf G^{(e)}\succeq 0$ then we can construct a knockoff matrix $\tilde{\mathbf X}$, where with
$\hat{\mathbf X}^{(e)} \equiv [\mathbf X^{(e)} \ \tilde{\mathbf X}\ \frac{1}{\|\tilde{\mathbf y}\|_2}\tilde{\mathbf y} \ \frac{1}{\sqrt{n}}\mathbf 1]$,
${\hat{\mathbf X}}^{(e)\top} \hat{\mathbf X}^{(e)} = \mathbf G^{(e)}$.
Equivalently, this ensures that the knockoff matrix $\tilde{\mathbf X}$ satisfies the construction according to Condition~\ref{ass:gram} along with two additional constraints,
$\mathbf X^\top \tilde{\mathbf y} = \tilde{\mathbf X}^\top \tilde{\mathbf y}$ and $\mathbf X^\top \mathbf 1 = \tilde{\mathbf X}^\top \mathbf 1$.

With the knockoffs constructed as above, we continue by applying the remaining steps of the signed-knockoff filter (Algorithm~\ref{alg:signed-kf}).
Our approach clearly violates the canonical knockoff-construction where the knockoffs are constructed \emph{without} looking at the response.
The knockoff-construction obeying Condition~\ref{ass:correct} is precisely what corrects for this violation: it requires that each feature and its knockoff have identical inner products with $\tilde{\mathbf{y}}$ and the intercept $\mathbf{1}$.
The constraint involving $\mathbf{1}$ is a technical requirement of the proof of Theorem~\ref{theorem:fdr_control} below, needed because the noise level $s_y^2$ depends on the intercept through the sample mean $\mathbf 1^{\!\top}\mathbf{y}/n$. 
These constraints negate any advantage the original features might have had over their knockoffs from being correlated with $\tilde{\mathbf y}$ or the intercept $\mathbf 1$.

Although Condition~\ref{ass:correct} amounts to Condition~\ref{ass:gram} together with the two additional constraints $\mathbf X^{\!\top}\tilde{\mathbf y} = \tilde{\mathbf X}^{\!\top}\tilde{\mathbf y}$ and $\mathbf X^{\!\top}\mathbf 1 = \tilde{\mathbf X}^{\!\top}\mathbf 1$, and so might be expected to require a larger sample size, a non-trivial $\mathbf s$ can still be obtained as soon as $n > p + 2$. 
The reason is that $\mathbf s$ need not be nonzero on every feature: if $\mathbf s$ has $m \equiv \lvert \{j : s_j > 0\} \rvert$ nonzero entries, it suffices that $n \geq p + 2 + m$, where the right-hand side accounts for the $p$ features, the $m$ nontrivial\footnote{Knockoffs $\tilde{\mathbf{X}}_j$ for which $\mathbf X_j \neq \tilde{\mathbf{X}}_j$.} knockoffs, and the two extra dimensions contributed by the constraints on $\tilde{\mathbf y}$ and $\mathbf 1$.
Since the number of nonzero entries is ours to choose, whenever $n > p + 2$ we may take any $1 \leq m \leq n - p - 2$ using a suitable strategy according to $\tilde{\mathbf y}$ (as we describe in Section~\ref{sec:choose-s}).
The standard knockoff filter, by contrast, requires $n \geq 2p$, and has no means of systematically reducing the knockoff-construction when $n < 2p$.

\begin{algorithm}[H]
	\caption{ {\bf Response-guided knockoff filter}}
	\label{alg:response-kf}
	\begin{algorithmic}[1]
		\Require  \begin{tabular}[t]{p{0.8\textwidth}}
		$\mathbf y \in \mathbb R^n$ the response vector; $\mathbf X \in \mathbb R^{n \times p}$ the design matrix with $n > p + 2$;  $d \in \{\pm 1\}$ the prespecified target sign of interest; $\alpha \in (0, 1)$ the prespecified directional FDR threshold \\
		\end{tabular}
		\Ensure A discovery list $\mathcal{\hat S}_{\text{RGKF}}$
        \State $\tilde{\mathbf y} \gets \mathbf y + \tilde{\boldsymbol{\varepsilon}}$ where $\tilde{\boldsymbol{\varepsilon}} \mid \mathbf y \sim \mathcal N(\mathbf 0, s_y^2 \mathbf I_n)$ \Comment{$s_y^2$ is the sample variance of $\mathbf y$}
        \State $\mathbf s \gets s(\mathbf X, \tilde{\mathbf y}, d) \in [0,1]^p$ \Comment{$\mathbf s$ satisfies Condition~\ref{ass:correct}}
        \State $\tilde{\mathbf X} \gets g(\mathbf X, \tilde{\mathbf y}, \mathbf s)$ 
			\Comment{$\tilde{\mathbf X}$ satisfies $\mathbf G^{(e)} = {\hat{\mathbf X}}^{(e)\top} \hat{\mathbf X}^{(e)}$}
        \State $\hat{\mathbf X} \gets [\mathbf X,  \tilde{\mathbf X}]$
        \State $\mathbf W \gets \mathbf W(\hat{\mathbf X}, \mathbf y)$ \Comment{$\mathbf W$ satisfies Conditions~\ref{ass:suff} and \ref{ass:sym}}
        \State $S_j \gets \sign\bigl((\mathbf X_{j}-\widetilde{\mathbf X}_{j})^{\!\top}\mathbf y\bigr) \cdot \sign\left(W_j \right)$ for $j = 1, \ldots, p$
        \State $\mathcal J_d \gets \{ j \in [p] : S_j = d  \}$
        \State \Return $\mathcal{\hat S}_{\text{RGKF}} \gets \text{SSS+}((W_j)_{j \in \mathcal J_d}, 1/2, \alpha)$ \Comment{Algorithm~\ref{alg:SeqStep}}
		\end{algorithmic}
\end{algorithm}

The response-guided knockoff filter inherently makes a trade-off: part of $\mathbf y$ is budgeted (in the form of $\tilde{\mathbf y}$) to inform the separation parameters $\mathbf s$, and the remaining information in $\mathbf y$ is used to determine the scores and signs, $\mathbf W$ and $\mathbf S$, respectively.
A pre-specified noise level for the perturbation $\tilde{\boldsymbol{\varepsilon}}$ cannot balance this trade-off effectively since it depends on the unknown regression variance $\sigma^2$ and the signal strength.
We therefore set the perturbation noise level to the sample variance $s_y^2$ as a data-adaptive perturbation scale that reflects both signal and noise.
The perturbation level cannot, however, be an arbitrary function of the data: just as $s_y^2$ requires the constraint involving $\mathbf 1$ in Condition~\ref{ass:correct}, other data-dependent choices would require further constraints (see the Discussion for more details).
We evaluate the choice of $s_y^2$ through a range of simulations and real-data experiments in Section~\ref{section:results}.
Finally, we establish that the response-guided knockoff filter controls the prespecified directional FDR in Theorem~\ref{theorem:fdr_control} below.
\begin{Theorem}
    \label{theorem:fdr_control}
    Assume $\mathbf y \sim \mathcal{N}(\mathbf X \boldsymbol \beta, \sigma^2 \mathbf I_n)$ with $\sigma^2>0$ and $n > p + 2$.
    Then the response-guided knockoff filter $\mathcal{\hat S}_{\text{RGKF}}$ controls the prespecified directional FDR as defined in~\eqref{eq:fdr_dir_d} at level $\alpha$.
\end{Theorem}
The proof machinery relies on several key Lemmas leveraging the framework of \cite{barber2019knockoff} and \cite{freestone2024semi}; hence we defer the full proof to Appendix~\ref{appendix:proofs}.
The proof also extends to feature-specific target signs $d_j \in \{\pm 1\}$ fixed in advance, yielding prespecified directional FDR control for those claims.
In this case, $\mathcal J_d$ in Step~7 of Algorithm~\ref{alg:response-kf} is replaced with $\{ j \in [p] : S_j = d_j  \}$.


\subsection{Related literature}
\citet{pournaderi2024variable} develop knockoff procedures for testing composite null hypotheses.
Their one-sided shifted ordinary least-squares (S-OLS) procedure tests $\beta_j\leq\delta_j$ against $\beta_j>\delta_j$, where $\delta_j\geq0$ is a prespecified effect-size threshold.
When $\delta_j=0$ for all $j$, this procedure coincides with the signed-knockoff filter with $d=+1$, using their OLS-based statistic.
Our formulation generalises this case by allowing any statistics satisfying Conditions~\ref{ass:suff} and~\ref{ass:sym}.
In particular, our formulation permits the lasso signed-max (LSM) statistic used in our experiments, for which Barber and Candès (2015) report good empirical performance.

Our response-guided knockoff filter works by partitioning the information in $\mathbf y$ into two parts: (a) a perturbed response $\tilde{\mathbf y}$ that is used to guide the knockoff-construction, and (b) the remaining information in $\mathbf y$ that is used for inference with prespecified directional FDR control.
This approach is similar in spirit to the work of \cite{barber2019knockoff} who directly partition the input data into two sets: one for reducing the set of candidate features and the other for selecting the relevant ones while controlling the FDR.
Notably their approach does not achieve finite-sample FDR control unless we model $\mathbf X$ as multivariate Gaussian, or we condition on the event that all relevant features in $\mathbf X$ have been selected in the first part.

Other partition approaches for false discovery rate control have been explored previously \citep{xing2023controlling,dai:false2,dai2023scale}.
However, like the approach above, these methods employ a hard data-splitting strategy, and use it to bypass knockoff-construction entirely rather than to inform it.
Furthermore, unlike the response-guided knockoff filter, these methods achieve asymptotic FDR control instead of finite-sample control.

Our approach to perturbing the response $\mathbf y$ is an instance of \emph{data fission}~\citep{leiner2025data}, a general framework in which a perturbation to the data is used for model selection, followed by a second inference step which uses the remaining information in the data conditioned on the perturbed data.
While Chapter 2 of \cite{leiner2025data} provides techniques for accomplishing data fission in many contexts, and in the linear regression setting, it is unclear how to accomplish this in the regression setting for the task of knockoff-construction.
Moreover, our procedure introduces a data-dependent perturbation, $\tilde{\boldsymbol{\varepsilon}} \mid \mathbf y$, to the response $\mathbf y$, which is not covered by \cite{leiner2025data} in the finite-sample setting.
Finally, we prove that Condition~\ref{ass:correct}, in addition to the other conditions,
is sufficient for achieving the second inference step, i.e., to perform feature selection with prespecified directional FDR control.

The response-guided knockoff filter is essentially a two-stage knockoff generator. 
The literature has proposed several complementary approaches to knockoff generation~\citep{barber:controlling,candes:panning,spector2022powerful,gimenez:improving}, which we benchmark against in the Results section using the signed-knockoff filter. 
However, a key distinction is that these generators are entirely response-agnostic: they do not use $\mathbf y$ when constructing knockoffs and therefore cannot explicitly prioritise features that are more likely to be relevant to the response or have the correct signs.

\subsection{Extension to FDX and k-FWER control}

Beyond prespecified directional FDR control, the response-guided knockoff filter extends naturally to other error criteria.
One such criterion is the \emph{false discovery exceedance} (FDX), which controls the probability that the FDP exceeds a target level~\citep{genovese:stochastic,genovese:exceedance,korn2004controlling}.
Specifically, a procedure controls the FDX at $(\alpha,\gamma)$ if
\[
\mathbb{P}\left( \operatorname{FDP} > \alpha \right) \leq \gamma.
\]
Another widely studied criterion is the \emph{$k$-family-wise error rate} ($k$-FWER)~\citep{lehmann:generalizations,vanderLaan:augmentation}, which controls the probability of making at least $k$ false discoveries, i.e.,
\[
\mathbb{P}\left( \lvert \mathcal H_0 \cap \hat{\mathcal S} \rvert \geq k \right) \leq \alpha.
\]

The prespecified directional FDX and $k$-FWER criteria are defined analogously by replacing $\mathcal H_0$ with $\mathcal H_0^{d}$ as per \eqref{eq:fdr_dir_d}.
The response-guided knockoff filter extends directly to these prespecified directional error criteria.
To obtain prespecified directional FDX or $k$-FWER control, it suffices to replace SSS+ in Step~8 of Algorithm~\ref{alg:response-kf} with the procedure that plays the same role in the ordinary knockoff filter: the FDX-controlling procedure of \citet{luo:competition} or the $k$-FWER-controlling procedure of \citet{janson:familywise} (see Appendix~\ref{appendix:algorithms} for the algorithms).
Both were shown to control their error rates under the special case of Assumption~\ref{assumption:general} in which the null signs are \emph{i.i.d.}\ with $c = 1/2$.
Using the techniques of \citet{freestone2024semi}, we show that their guarantees extend to the full Assumption~\ref{assumption:general} with independent null signs whose positive probabilities are at most $1/2$. 
The proof of Theorem~\ref{theorem:fdr_control} establishes that this assumption holds for the response-guided scores, yielding Theorem~\ref{theorem:fdx_fwer_control} below.

\begin{Theorem}
    \label{theorem:fdx_fwer_control}
    Assume $\mathbf y \sim \mathcal{N}(\mathbf X \boldsymbol \beta, \sigma^2 \mathbf I_n)$.
    Then Algorithms~\ref{alg:response-kf-fdx} and \ref{alg:response-kf-kfwer} control the prespecified directional FDX at level $(\alpha, \gamma)$ and the prespecified directional $k$-FWER at level $\alpha$, respectively.
\end{Theorem}

\subsection{A concrete implementation of the response-guided knockoff filter}
\label{sec:choose-s}

As described so far, the response-guided knockoff filter is a general framework:
its validity does not depend on a particular rule for choosing the separation
parameters $\mathbf s$, provided that Condition~\ref{ass:correct} holds.
We now describe the specific two-stage implementation used in our experiments.

\subsubsection{Stage 1: Screening} 
\label{sec:stage_1}
We first identify a subset $\hat{\mathcal S}_0 \subseteq [p]$ of features that is likely to contain the
target set of interest, $\mathcal H_1^{d} \equiv [p]\setminus \mathcal H_0^{d}$, by regressing $\tilde{\mathbf y}$ onto $\mathbf X$\footnote{We omit the intercept term, which is consistent with \citet{barber:controlling}.}.
We consider the following two approaches for constructing $\hat{\mathcal S}_0$.
\begin{itemize}
    \item The \textit{signed} approach defines the screened features as $\hat{\mathcal S}_0 \equiv \{j \in [p] : \boldsymbol{\hat{\beta}}_{\text{OLS}, j} \cdot d > 0 \}$, where $\boldsymbol{\hat{\beta}}_{\text{OLS}, j}$ denotes the ordinary least squares estimate of the regression coefficients.
    \item The \textit{regularise-signed} approach defines the screened features as $\hat{\mathcal S}_0 \equiv \{j \in [p] : \boldsymbol{\hat{\beta}}_{\text{enet}, j} \cdot d > 0 \}$, where $\boldsymbol{\hat{\beta}}_{\text{enet}, j}$ denotes the elastic-net estimate of the regression coefficients~\citep{zou:regularization}. 
    For the elastic-net fit, we tune both the mixing parameter
	$\alpha_{\text{enet}}$ and the penalty parameter $\lambda_{\text{enet}}$ by
	cross-validation. The mixing parameter $\alpha_{\text{enet}}$ controls the
	relative weight of the $\ell_1$ and $\ell_2$ penalties, while
	$\lambda_{\text{enet}}$ controls the overall amount of regularisation.
	We run 10-fold cross-validation using \texttt{cv.glmnet}, with no
	standardisation (features already have unit length), jointly over the default \texttt{glmnet} grid of
	$\lambda_{\text{enet}}$ values and $\alpha_{\text{enet}} \in \{ 0.2, 0.4, 0.6, 0.8 \}$.  We select the value of
	$\alpha_{\text{enet}}$ and $\lambda_{\text{enet}}$ that attains the smallest mean
	cross-validated error. In the simulation experiments, if the resulting
	regularised-signed screen is empty, we use the signed OLS screen for that
	replicate as a fallback; the replicate remains included under the
	\texttt{RG(reg-signed,ME)} label in the aggregated results.
\end{itemize}
Recall from Section~\ref{section:response-guided_kf} that for $n > p + 2$ observations, we can choose up to $n-p-2$ features to construct nontrivial knockoffs.
Thus if $\lvert \hat{\mathcal S}_0 \rvert$ exceeds this, $\hat{\mathcal S}_0$ is reduced to the top $n - p - 2$ features ranked according to $\boldsymbol{\hat{\beta}}_{\text{OLS}} \cdot d$ or $\boldsymbol{\hat{\beta}}_{\text{enet}} \cdot d$, respectively.
Features outside $\hat{\mathcal S}_0$ are assigned $s_j = 0$ in Stage~2, so that $\tilde{\mathbf X}_j = \mathbf X_j$.
Swapping such a feature with its identical knockoff leaves the design unchanged, so anti-symmetry (Condition~\ref{ass:sym}) forces $W_j = 0$, and the feature is never selected.

\subsubsection{Stage 2: knockoff-construction on the screened set}
For the retained features $j\in \hat{\mathcal S}_0$, we choose $ (s_j)_{j\in\hat{\mathcal S}_0}$ using a modification of the maximum-entropy (ME) construction of \citep{gimenez:improving,spector2022powerful}, adapted to the fact that only a subset of coordinates are given nonzero separation parameters.
For any candidate $(s_j)_{j\in\hat{\mathcal S}_0}$, extend it to $\mathbf s$ by setting $s_j=0$ for $j\in[p]\setminus\hat{\mathcal S}_0$, and let
$\mathbf D^{(e)}(\mathbf s)\equiv\operatorname{diag}\{\mathbf s,0,0\}\in\mathbb R^{(p+2)\times(p+2)}$.
Define the reduced augmented Gram matrix
\[
\mathbf G_{\hat{\mathcal S}_0}^{(e)}(\mathbf s)
\equiv
\begin{bmatrix}
\boldsymbol{\Sigma}^{(e)}
&
\bigl(\boldsymbol{\Sigma}^{(e)}-\mathbf D^{(e)}(\mathbf s)\bigr)_{\bullet,\hat{\mathcal S}_0}
\\
\bigl(\boldsymbol{\Sigma}^{(e)}-\mathbf D^{(e)}(\mathbf s)\bigr)_{\hat{\mathcal S}_0,\bullet}
&
\boldsymbol{\Sigma}^{(e)}_{\hat{\mathcal S}_0,\hat{\mathcal S}_0}
\end{bmatrix},
\]
where $\bullet$ in a matrix subscript denotes all $p+2$ augmented indices.
This is the target Gram matrix for $[\,\mathbf X^{(e)}\ \widetilde{\mathbf X}_{\hat{\mathcal S}_0}\,]$, where $\widetilde{\mathbf X}_{\hat{\mathcal S}_0}$ contains the knockoff columns indexed by $\hat{\mathcal S}_0$; it retains every column of the augmented design and only the knockoff columns corresponding to the screened features.
We then solve
\[
\mathbf s^{\star}_{\hat{\mathcal S}_0}
\in \argmax_{(s_j)_{j\in\hat{\mathcal S}_0}\in[0,1]^{|\hat{\mathcal S}_0|}}
\left\{
\det\!\bigl(\mathbf G_{\hat{\mathcal S}_0}^{(e)}(\mathbf s)\bigr)
:
\mathbf G_{\hat{\mathcal S}_0}^{(e)}(\mathbf s)\succeq0
\right\},
\]
where we assume that the augmented design $\mathbf X^{(e)}$ has full column rank, or equivalently that $\boldsymbol{\Sigma}^{(e)}\succ0$ as otherwise $\mathbf G_{\hat{\mathcal S}_0}^{(e)}(\mathbf s)$ is singular for every feasible $\mathbf s$.
The full extended design associated with \(\mathbf G^{(e)}\) differs from \([\,\mathbf X^{(e)}\ \widetilde{\mathbf X}_{\hat{\mathcal S}_0}\,]\) only by duplicate columns, so \(\mathbf G^{(e)}\succeq0\) if and only if \(\mathbf G_{\hat{\mathcal S}_0}^{(e)}(\mathbf s)\succeq0\), so the constraint above is precisely Condition~\ref{ass:correct}.
We maximise the determinant of the reduced Gram matrix because the duplicate columns make $\det(\mathbf G^{(e)})=0$.
After obtaining $\mathbf s^{\star}_{\hat{\mathcal S}_0}$, extend it by zeros outside $\hat{\mathcal S}_0$ to obtain $\mathbf s^\star$.
We construct the active knockoffs so that the Gram matrix of $[\,\mathbf X^{(e)}\ \widetilde{\mathbf X}_{\hat{\mathcal S}_0}\,]$ equals $\mathbf G_{\hat{\mathcal S}_0}^{(e)}(\mathbf s^\star)$.

The maximum-entropy objective above can be replaced with any standard fixed-$X$ knockoff objective, with the same modification that nonzero separation is permitted only on the screened set $\hat{\mathcal S}_0$. 
For example, one may use SDP knockoffs \citep{barber:controlling} or equi-correlated knockoffs (Section~\ref{section:overview_FX}), with the latter admitting a closed-form solution (see Appendix~\ref{appendix:proofs}). 

\section{Results}
\label{section:results}

We consider variants of the response-guided knockoff filter, corresponding to different choices of the screening step used to construct $\hat{\mathcal S}_0$.
To keep the exposition concise, we denote each variant by \texttt{RG(x,y)}, where \texttt{x} specifies the screening approach (Stage~1) and \texttt{y} specifies the knockoff
construction applied on $\hat{\mathcal S}_0$ (Stage~2).
For example, \texttt{RG(reg-signed,ME)} uses the regularised-signed screening approach for constructing $\hat{\mathcal S}_0$ and the maximum-entropy construction on $\hat{\mathcal S}_0$.
We similarly denote the signed-knockoff filter of Section~\ref{section:signed-knockoff-filter} as \texttt{S(y)}, where \texttt{y} specifies the knockoff-construction used.
Table~\ref{table:naming_conventions} summarises these naming conventions, together with the screening rules \texttt{x}, knockoff-constructions \texttt{y}, and feature statistic $ W$ used throughout this section.

\begin{table}
    \caption{\label{table:naming_conventions}\textbf{Naming conventions for the methods compared in Section~\ref{section:results}.}
    Each variant is written as \texttt{RG(x,y)} or \texttt{S(y)}, where \texttt{x} is the screening rule and \texttt{y} is the knockoff-construction (e.g., \texttt{RG(reg-signed,ME)} or \texttt{S(MVR)}).}
    \centering
    \small
    \setlength{\tabcolsep}{6pt}%
    \renewcommand{\arraystretch}{1.1}%
    \begin{tabularx}{\linewidth}{l l >{\raggedright\arraybackslash}X}
        \toprule
        & Abbreviation & Meaning \\
        \midrule
        \multirow{2}{*}{Filter}
          & \texttt{RG(x,y)} & Response-guided knockoff filter (Section~\ref{section:response-guided_kf}), with screening rule \texttt{x} (Stage~1) and knockoff-construction \texttt{y} (Stage~2) \\
          & \texttt{S(y)}    & Signed-knockoff filter (Section~\ref{section:signed-knockoff-filter}), with knockoff-construction \texttt{y} \\
        \midrule
        \multirow{2}{*}{Screening \texttt{x}}
          & \texttt{signed}     & Signed screening using the ordinary least squares (OLS) estimate: $\hat{\mathcal S}_0 = \{j \in [p] : \boldsymbol{\hat{\beta}}_{\text{OLS}} \cdot d > 0\}$ \\
          & \texttt{reg-signed} & Regularise-signed screening using the elastic-net (enet) estimate~\citep{zou:regularization}: $\hat{\mathcal S}_0 = \{j \in [p] : \boldsymbol{\hat{\beta}}_{\text{enet}} \cdot d > 0\}$ \\
        \midrule
        \multirow{4}{*}{Construction \texttt{y}}
          & \texttt{equi} & Equi-correlated knockoffs (Section~\ref{section:overview_FX}) \\
          & \texttt{SDP}  & Semidefinite-programming knockoffs~\citep{barber:controlling} \\
          & \texttt{ME}   & Maximum-entropy knockoffs~\citep{gimenez:improving,spector2022powerful} \\
          & \texttt{MVR}  & Minimum variance-based reconstructability knockoffs~\citep{spector2022powerful} \\
        \midrule
        Statistic $W$
          & \texttt{LSM} & Lasso signed-max feature statistic~\citep{barber:controlling} \\
        \bottomrule
    \end{tabularx}%
\end{table}

Even under the fixed-X framework, 
a selected knockoff construction involves arbitrary choices (such as an orthonormal basis).
To avoid results depending on a particular arbitrary choice, we use a \emph{randomised} fixed-X knockoff-construction.
Our implementations of the ME and minimum variance-based reconstructability (MVR) constructions follow \citet{spector2022powerful}, with the additional capability of randomising the fixed-X knockoff-construction through the choice of orthonormal basis.
Equi-correlated and SDP knockoffs were constructed using the \texttt{knockoff} R package with randomisation enabled.
In the analyses of Sections~\ref{section:simulation_1}, \ref{section:simulation_2}, and~\ref{section:hiv_analysis}, all methods were implemented in R.

We used the lasso signed max (LSM) score~\citep{barber:controlling} to compute the knockoff statistics $\mathbf W$ in all simulations.
Specifically, let $Z_j$ and $\tilde{Z}_j$ denote the value of the regularisation parameter $\lambda$ at which feature $j$ and its knockoff first enter the lasso path, respectively.
Then the LSM statistic is defined as 
\[W_j \equiv \max(Z_j, \tilde{Z}_j) \cdot \sign(Z_j - \tilde{Z}_j).\]
For the response-guided knockoff filter, the lasso path is fit on all $p$ original features together with the knockoffs of the screened features $\hat{\mathcal S}_0$. 
The LSM statistic $W_j$ is subsequently computed only for $j \in \hat{\mathcal S}_0$, with $W_j = 0$ for all features $j$ outside $\hat{\mathcal S}_0$ as outlined in Section~\ref{sec:stage_1}, so that such features act as covariates in the lasso fit but are never selected.

Throughout the simulations below, it suffices to consider the case where $d = +1$, i.e., when the target sign of interest is positive. 
Clearly, the case when $d = -1$ can be inferred by inverting the signs of $\boldsymbol{\beta}$.

\subsection{Simulation 1: \cite{spector2022powerful}}
\label{section:simulation_1}

To illustrate the effectiveness of the response-guided knockoff filter, we simulated a diverse range of linear models according to the fixed-X simulation setup of \cite{spector2022powerful}. 
We generate simulated data according to six different types of linear models which vary according to the correlation structure of the features and the true coefficients $\boldsymbol \beta$. 
The six types are denoted as \texttt{AR1}, \texttt{AR1 (Corr)}, \texttt{Block Equi}, \texttt{ER (Cov)}, \texttt{ER (Prec)}, \texttt{Equi}, with description of each type provided in Appendix~\ref{appendix:sim_details}.
For each type, we generate $p = 500$ features and generate 100 replicates for each combination of $n$ and $k$ where
\begin{itemize}
    \item $n \in \{1005, 1250, 1500, 1750, 2000\}$ is the number of observations (notably $n \ge 2p$);
    \item $k \in \{50, 100, 150\}$ is the number of nonzero coefficients,
\end{itemize}
and where the design $\mathbf X$, the coefficients $\boldsymbol \beta$, and the response $\mathbf y$, together with the covariance $\boldsymbol{\Sigma}$ generating $\mathbf X$ for the six model types in which $\boldsymbol{\Sigma}$ is random, are redrawn for each replicate.
We define the empirical power and empirical prespecified directional FDR as
\begin{align}
    \label{eq:empirical_power_fdr}
    \widehat{\mathrm{Power}} \equiv \frac{1}{100} \sum_{r=1}^{100} \frac{|\mathcal{\hat S}^{(r)} \cap \mathcal H_1^{d}|}{|\mathcal H_1^{d}|}, \quad
    \widehat{\mathrm{FDR}}_{\mathrm{dir}}^{d} \equiv \frac{1}{100} \sum_{r=1}^{100} \frac{|\mathcal{\hat S}^{(r)}   \cap \mathcal H_0^{d}|}{|\mathcal{\hat S}^{(r)}| \vee 1},
\end{align}
where $\mathcal{\hat S}^{(r)}$ denotes the selected features in replicate $r$ and $\mathcal{H}_1^d \equiv [p] \setminus \mathcal{H}_0^d$. 
The theory controls the prespecified directional FDR conditional on a fixed design $\mathbf X$. We redraw the design in each replicate so that the reported power and empirical prespecified directional FDR describe performance not specific to any realised design.
We compare the response-guided knockoff filter variants (Section~\ref{sec:choose-s}) against the signed-knockoff filter (Section~\ref{section:signed-knockoff-filter}) using the equi-correlated, SDP, ME, and MVR knockoff-constructions.
To place the methods on a common scale, we report a \emph{relative power} score: for each prespecified directional FDR level $\alpha$, $n$, $k$, and each of the six linear-model types, we divide $\widehat{\text{Power}}$ given by each method by the largest $\widehat{\text{Power}}$ made by any method in that instance, so that the best method in each instance scores $1$.
This yields $4 \times 5 \times 3 \times 6 = 360$ scores per method.
The first column of Figure~\ref{fig:tpr_summary} summarises these scores: the top panel shows the distribution of a method's $360$ scores as a boxplot, and the bottom panel shows their average.
The boxplots of the two response-guided knockoff filter variants are similar to one another and show that overall they are the most powerful methods.
The signed-knockoff filter variants using \texttt{ME} and \texttt{MVR} are more powerful than those using \texttt{equi} and \texttt{SDP}, which is consistent with the results in \cite{spector2022powerful}.
The per-type power plots, in which the empirical power is averaged over $n$ and $k$ within each linear-model type, are shown in Figure~\ref{fig:spector_simulation_power} (Appendix~\ref{appendix:power_plots}).
The complementary empirical prespecified directional FDR plot (Figure~\ref{fig:spector_simulation_fdr}) reports a single setting ($n = 1500$, $k = 100$) for each linear-model type.
Since the design is redrawn in each replicate, the plot gives an impression of each method's empirical prespecified directional FDR over a distribution of designs, rather than for any particular design.

\begin{figure}[h!]
    \centering
    \includegraphics[width=6in]{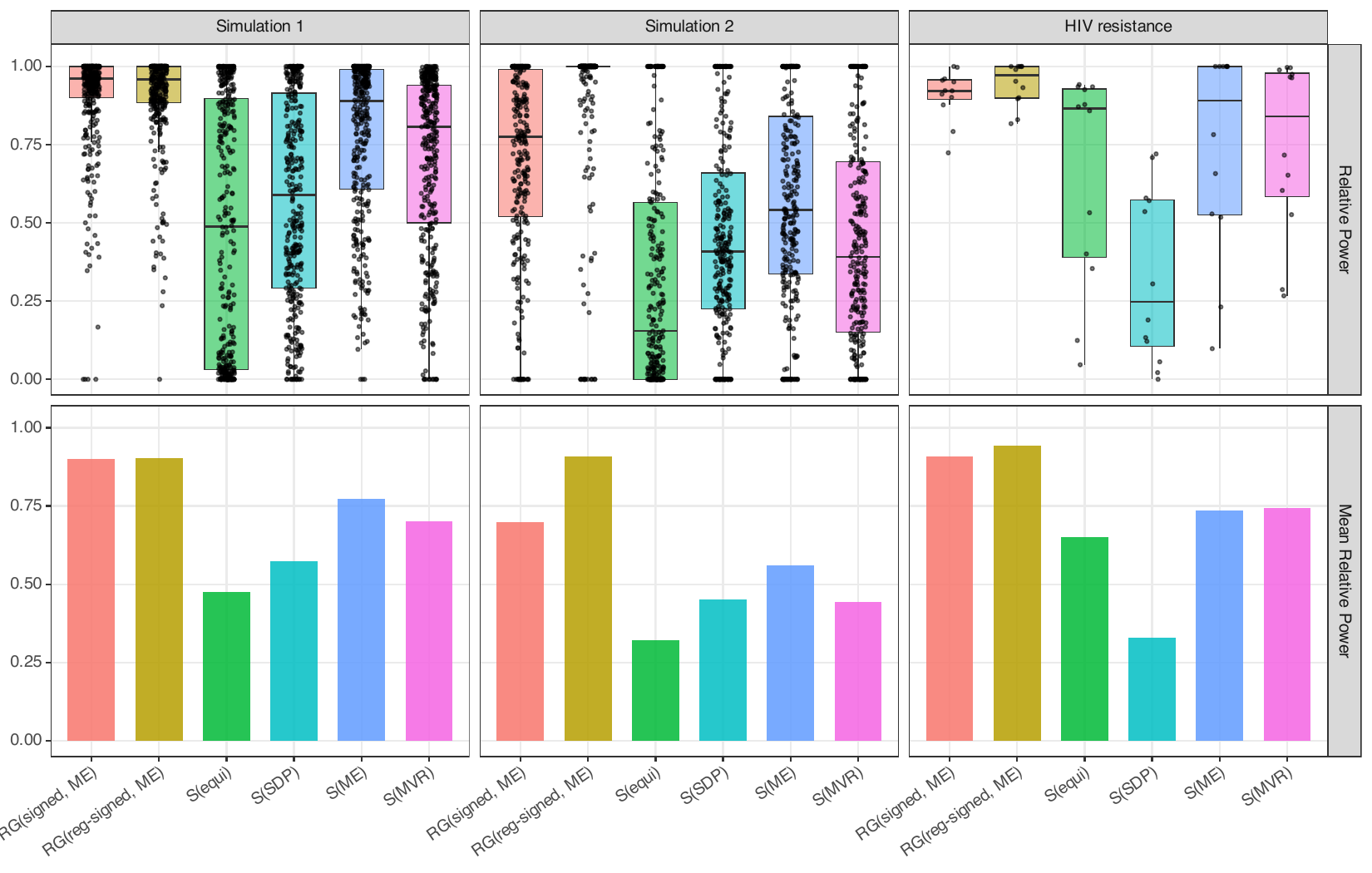}
    \caption{\textbf{Relative power of the response-guided and signed-knockoff filters across the three experiments.}
    For each prespecified directional FDR level $\alpha \in \{0.05, 0.1, 0.2, 0.3\}$, sample size $n$, sparsity level $k$, and linear-model type (Simulations~1 and~2), or for each prespecified directional FDR level and PI drug (HIV data), each method's $\widehat{\text{Power}}$ is divided by the largest $\widehat{\text{Power}}$ made by any method in that instance.
    The top row shows the distribution of these relative-power scores for each method as a boxplot ($360$, $288$, and $12$ scores per method for Simulation~1, Simulation~2, and the HIV data, respectively), and the bottom row shows their average.
    The three columns correspond to Simulation~1 ($n \ge 2p$), Simulation~2 ($p + 2 < n < 2p$), and the HIV drug-resistance data.}
    \label{fig:tpr_summary}
\end{figure}

\subsection{Simulation 2: $p + 2 < n < 2p$}
\label{section:simulation_2}
The standard knockoff filter, and therefore the signed-knockoff filter, requires that $n \geq 2p$ to construct valid knockoffs~\citep{barber:controlling}.
When $p < n < 2p$, existing fixed-X knockoff generators rely on a heuristic that extends the data to $2p$ observations~\citep{barber:controlling,spector2022powerful}.
Specifically, $2p - n$ rows of zeros are appended to $\mathbf X$, and $2p - n$ additional responses, drawn from the fitted OLS model at these rows, i.e., independently from $\mathcal N(0, \hat{\sigma}^2_{\text{OLS}})$ where $\hat{\sigma}^2_{\text{OLS}}$ is the OLS estimate of $\sigma^2$, are appended to $\mathbf y$.
The signed-knockoff filter is then applied to the extended data by Algorithm~\ref{alg:signed-kf}.
Because $\hat{\sigma}^2_{\text{OLS}}$ only estimates $\sigma^2$, the extended data do not exactly follow the linear model, so the resulting signed-knockoff filter no longer guarantees finite-sample prespecified directional FDR control.
In contrast, the response-guided knockoff filter can be implemented directly when $p + 2 < n < 2p$ (see Section~\ref{sec:choose-s}), with prespecified directional FDR control guaranteed by Theorem~\ref{theorem:fdr_control}, and, as we show below, it is typically more powerful in this setting.

To evaluate the performance of the response-guided knockoff filter in this setting, we consider a similar simulation setup to Simulation~1, but with $n \in \{600, 700, 800, 900\}$.
We compare our results to the signed-knockoff filter using the above heuristic of generating additional response observations.
We summarise performance with the same relative-power score as in Simulation~1, again yielding $4 \times 4 \times 3 \times 6 = 288$ scores per method (four prespecified directional FDR levels $\alpha$, four sample sizes $n$, three number of nonzero coefficients $k$, and six linear-model types).
As summarised in the second column of Figure~\ref{fig:tpr_summary}, the response-guided knockoff filter variants typically outperform the signed-knockoff filter, often substantially so, with the \texttt{RG(reg-signed, ME)} variant clearly performing best.
A plausible explanation is that because the elastic-net (\texttt{reg-signed}) screening prunes more aggressively than the OLS-based (\texttt{signed}) screening, it concentrates this limited budget on fewer, more promising features and increases the separation parameters $s_j$ on the retained features.
The flip side is that it screens out more features that are in $\mathcal{H}_1^{d}$; however, such marginal signals are difficult to detect at all when $n$ is close to $p$, so the gain from a smaller $\hat{\mathcal S}_0$ outweighs this loss.
The per-type power curves, again averaged over $n$ and $k$ within each linear-model type, are shown in Figure~\ref{fig:spector_simulation_2_power} (Appendix~\ref{appendix:power_plots}).
The complementary empirical prespecified directional FDR plot (Figure~\ref{fig:spector_simulation_2_fdr}) again reports a single setting ($n = 700$, $k = 100$).

\subsection{HIV drug resistance data}
\label{section:hiv_analysis}

We evaluate the response-guided knockoff filter on a Human Immunodeficiency Virus (HIV) drug resistance study \citep{rhee:genotypic} following a similar setup to \cite{ren:derandomised}. 
Specifically, we downloaded the HIV drug resistance data from the Stanford HIV Drug Resistance Database.
We considered the three main protease inhibitor (PI) drugs currently in use: Atazanavir (ATV), Lopinavir (LPV) and Darunavir (DRV).
A ground truth was approximated by downloading the list of mutations from \url{https://hivdb.stanford.edu/dr-summary/comments/PI/}.
Since this list contains both mutations that increase and decrease drug resistance, we defined an approximate ground truth as those mutations that did not clearly indicate a decrease in resistance for that drug.
Further details are given in Appendix~\ref{appendix:hiv_details}.

For each method we recorded the average number of approximate ground-truth discoveries over 50 replications, where in each replicate we randomly drew a fixed-X knockoff design matrix.
We summarise performance with the same relative-power score as in Simulation~1 and 2, yielding $4 \times 3 = 12$ scores per method (four prespecified directional FDR levels and three drugs).
As summarised in the third column of Figure~\ref{fig:tpr_summary}, the response-guided knockoff filters are the strongest overall performer across the three PI drugs.
The per-drug discovery plots in Figure~\ref{fig:hiv_power} (Appendix~\ref{appendix:power_plots}) show that \texttt{S(ME)}, \texttt{S(MVR)} and \texttt{S(equi)} perform comparably to \texttt{RG(signed, ME)} at higher prespecified directional FDR thresholds, but are outperformed by the response-guided knockoff filter variants at lower thresholds.

\section{Discussion}
\label{section:discussion}

We have presented the response-guided knockoff filter, a selective inference framework that incorporates response information into knockoff-construction while provably controlling the prespecified directional FDR.
Existing knockoff generators are agnostic to the response as this ensures that the knockoff statistics satisfy the necessary properties for FDR control~\citep{barber:controlling,candes:panning}.
In contrast, our response-guided knockoff filter leverages part of the information in the response to guide the knockoff-construction. 
Because $\tilde{\mathbf y}$ itself is derived from $\mathbf y $ and the full response is later reused to compute the knockoff statistics, this procedure may superficially look like ``double dipping''. 
However, the crucial point is that Condition~\ref{ass:correct} and the corresponding knockoff-construction neutralise this problem: it forces each original feature and its knockoff to have identical correlations with the guiding response and the intercept, so the response-guided construction does not give original features an unfair advantage in the inference step. 
Our empirical results demonstrate that the response-guided knockoff filter can lead to substantial power gains over existing knockoff generators.

The response-guided knockoff filter also offers some additional benefits.
First, the response-guided knockoff filter can be implemented when $p + 2 < n < 2p$, while all existing fixed-X knockoff generators require $n \geq 2p$ with heuristics implemented when $n < 2p$ that no longer guarantee finite-sample FDR control~\citep{barber:controlling}.
Our empirical results show that despite this fact, response-guided knockoffs are even more powerful than the signed knockoff filter that relies on this heuristic.
Second, our framework controls other error rates such as the prespecified directional FDX and $k$-FWER.

As outlined in Section~\ref{sec:choose-s}, our implementation of the response-guided knockoff filter uses a two-stage procedure to construct the separation parameters $\mathbf{s}$. 
While this strategy performed well empirically, it is ultimately a heuristic. 
In particular, our decision to use elastic-net (rather than the lasso~\citep{tibshirani:regression} or other alternatives) when forming $\hat{\mathcal S}_0$ was motivated by the well-known behaviour of the lasso under strong collinearity: when predictors are highly correlated, the lasso often selects only one variable from a correlated group, potentially discarding other relevant features. 
By contrast, the elastic net tends to exhibit a grouping effect, encouraging correlated predictors to enter (or leave) the model together~\citep{zou:regularization}, ensuring $\hat{\mathcal{S}}_0$ does not accidentally exclude relevant features.
More broadly, a hard screen is only one way of using $\tilde{\mathbf y}$ as Condition~\ref{ass:correct} permits any rule for choosing $\mathbf s$ from $\mathbf X$ and $\tilde{\mathbf y}$.
For example, rather than setting $s_j = 0$ outside $\hat{\mathcal S}_0$, one could spread the separation budget incrementally across features according to the evidence in $\tilde{\mathbf y}$ that each has the target sign.
A more theoretically grounded rule for constructing $\mathbf s$ along these lines could yield further power gains.

Our proof machinery in Appendix~\ref{appendix:proofs} heavily relies on the fact that the data is a linear model.
In contrast, the \textit{model-X} knockoff framework of \cite{candes:panning} allows for an arbitrary conditional distribution of $ Y \mid  X$, while assuming that the distribution of $ X$ is known, or at least well approximated~\citep{barber2020robust,fan2025asymptotic}.
Hence, our response-guided knockoff filter cannot be easily adapted to the model-X assumptions while maintaining its finite-sample error control.
Therefore, an important direction for future work is to explore alternative ways for constructing response-guided knockoffs under the model-X assumptions.

Another avenue of research is refining how the size of the perturbation $\tilde{\boldsymbol{\varepsilon}}$ is chosen.
This size determines how the information in $\mathbf y$ is divided between guiding the knockoff construction and the subsequent inference.
A smaller perturbation makes $\tilde{\mathbf y}$ a sharper guide, giving a more accurate screened set $\hat{\mathcal S}_0$ and better-targeted separation parameters $\mathbf s$.
However, because Condition~\ref{ass:correct} makes the inference invariant to $\tilde{\mathbf y}$, a smaller perturbation also leaves less signal for the scores $\mathbf W$.
A larger perturbation has the opposite effect.
Our choice, the sample variance $s_y^2$ of $\mathbf y$, already adapts to the data, but only through the overall scale of $\mathbf y$, and need not be optimal.
A more adaptive choice may lead to further power gains, but it must be accounted for in the knockoff construction, just as the constraint involving $\mathbf 1$ in Condition~\ref{ass:correct} accounts for $s_y^2$, and would therefore typically require further constraints.

In a different line of work, \cite{emery:multiple2} and \cite{gimenez:improving} have proposed the idea of multiple simultaneous knockoffs which can increase power, particularly when the FDR threshold, $\alpha$, is low.
Unfortunately these approaches are limited to the case when $n \geq (k + 1)\cdot p$, where $k$ is the number of knockoffs constructed per feature.
Since the response-guided knockoff filter reduces the number of features that require knockoffs to be constructed, it is possible to combine the two ideas to construct multiple response-guided knockoffs when $n$ is only moderately larger than $p$.
However, this modification is nontrivial as it is unclear how to define the signs, $\mathbf S$, when multiple knockoffs are constructed.
This modification is left for future work.

\subsection*{Funding}

Garth Tarr and Samuel Muller were supported by the Australian Research Council (ARC) under grant number DP260100348.

\subsection*{Conflict of interest}

None declared.

\subsection*{Data availability}

Documented scripts reproducing the numerical results of this article are available at \url{https://github.com/freejstone/response_guided_knockoffs_code}. 
All datasets analysed are publicly available, with the procedures for obtaining them described within the article.

\subsection*{Acknowledgements}

The first author used OpenAI's GPT-6 Astra model to assist with editting the manuscript and codebase.

\clearpage
\bibliographystyle{plainnat}
\bibliography{refs}

\begin{thebibliography}{28}
\providecommand{\natexlab}[1]{#1}
\providecommand{\url}[1]{\texttt{#1}}
\expandafter\ifx\csname urlstyle\endcsname\relax
  \providecommand{\doi}[1]{doi: #1}\else
  \providecommand{\doi}{doi: \begingroup \urlstyle{rm}\Url}\fi

\bibitem[Barber and Cand{\`e}s(2015)]{barber:controlling}
R.~F. Barber and Emmanuel~J. Cand{\`e}s.
\newblock Controlling the false discovery rate via knockoffs.
\newblock \emph{The Annals of Statistics}, 43\penalty0 (5):\penalty0
  2055--2085, 2015.

\bibitem[Barber and Cand{\`e}s(2019)]{barber2019knockoff}
R.~F. Barber and Emmanuel~J. Cand{\`e}s.
\newblock {A knockoff filter for high-dimensional selective inference}.
\newblock \emph{The Annals of Statistics}, 47\penalty0 (5):\penalty0 2504 --
  2537, 2019.

\bibitem[Barber et~al.(2020)Barber, Cand{\`e}s, and Samworth]{barber2020robust}
Rina~Foygel Barber, Emmanuel~J Cand{\`e}s, and Richard~J Samworth.
\newblock Robust inference with knockoffs.
\newblock \emph{The Annals of Statistics}, 48\penalty0 (3):\penalty0
  1409--1431, 2020.

\bibitem[Benjamini and Hochberg(1995)]{benjamini:controlling}
Y.~Benjamini and Y.~Hochberg.
\newblock Controlling the false discovery rate: a practical and powerful
  approach to multiple testing.
\newblock \emph{Journal of the Royal Statistical Society Series B},
  57:\penalty0 289--300, 1995.

\bibitem[Cand{\`e}s et~al.(2018)Cand{\`e}s, Fan, Janson, and
  Lv]{candes:panning}
E.~J. Cand{\`e}s, Y.~Fan, L.~Janson, and J.~Lv.
\newblock Panning for gold: Model-{X} knockoffs for high-dimensional controlled
  variable selection.
\newblock \emph{Journal of the Royal Statistical Society Series B}, 80\penalty0
  (3):\penalty0 551--577, 2018.

\bibitem[Dai et~al.(2023{\natexlab{a}})Dai, Lin, Xing, and Liu]{dai2023scale}
Chenguang Dai, Buyu Lin, Xin Xing, and Jun~S Liu.
\newblock A scale-free approach for false discovery rate control in generalized
  linear models.
\newblock \emph{Journal of the American Statistical Association}, 118\penalty0
  (543):\penalty0 1551--1565, 2023{\natexlab{a}}.

\bibitem[Dai et~al.(2023{\natexlab{b}})Dai, Lin, Xing, and Liu]{dai:false2}
Chenguang Dai, Buyu Lin, Xin Xing, and Jun~S Liu.
\newblock False discovery rate control via data splitting.
\newblock \emph{Journal of the American Statistical Association}, 118\penalty0
  (544):\penalty0 2503--2520, 2023{\natexlab{b}}.

\bibitem[Emery et~al.(2019)Emery, Hasam, Noble, and Keich]{emery:multiple2}
K.~Emery, S.~Hasam, W.~S. Noble, and U.~Keich.
\newblock Multiple competition-based {FDR} control for peptide detection.
\newblock \emph{arXiv}, 2019.
\newblock arXiv:1907.01458v2.

\bibitem[Fan et~al.(2025)Fan, Gao, Lv, and Xu]{fan2025asymptotic}
Yingying Fan, Lan Gao, Jinchi Lv, and Xiaocong Xu.
\newblock Asymptotic fdr control with model-x knockoffs: Is moments matching
  sufficient?
\newblock \emph{arXiv preprint arXiv:2502.05969}, 2025.

\bibitem[Freestone et~al.(2024)Freestone, Noble, and Keich]{freestone2024semi}
Jack Freestone, William~Stafford Noble, and Uri Keich.
\newblock A semi-supervised framework for diverse multiple hypothesis testing
  scenarios.
\newblock \emph{arXiv preprint arXiv:2411.15771}, 2024.

\bibitem[Genovese and Wasserman(2004)]{genovese:stochastic}
Christopher Genovese and Larry Wasserman.
\newblock A stochastic process approach to false discovery control.
\newblock \emph{The Annals of Statistics}, 32\penalty0 (3):\penalty0
  1035--1061, 06 2004.

\bibitem[Genovese and Wasserman(2006)]{genovese:exceedance}
CR~Genovese and L~Wasserman.
\newblock Exceedance control of the false discovery proportion.
\newblock \emph{Journal of the American Statistical Association}, 101\penalty0
  (476):\penalty0 1408--1417, 2006.

\bibitem[Gimenez and Zou(2019)]{gimenez:improving}
Jaime~Roquero Gimenez and James Zou.
\newblock Improving the stability of the knockoff procedure: Multiple
  simultaneous knockoffs and entropy maximization.
\newblock In \emph{The 22nd International Conference on Artificial Intelligence
  and Statistics}, pages 2184--2192. PMLR, 2019.

\bibitem[Janson and Su(2016)]{janson:familywise}
Lucas Janson and Weijie Su.
\newblock Familywise error rate control via knockoffs.
\newblock \emph{Electronic Journal of Statistics}, 10\penalty0 (1):\penalty0
  960--975, 2016.

\bibitem[Korn et~al.(2004)Korn, Troendle, McShane, and
  Simon]{korn2004controlling}
Edward~L Korn, James~F Troendle, Lisa~M McShane, and Richard Simon.
\newblock Controlling the number of false discoveries: application to
  high-dimensional genomic data.
\newblock \emph{Journal of Statistical Planning and Inference}, 124\penalty0
  (2):\penalty0 379--398, 2004.

\bibitem[Lehmann and Romano(2005)]{lehmann:generalizations}
E.~L. Lehmann and Joseph~P. Romano.
\newblock Generalizations of the familywise error rate.
\newblock \emph{The Annals of Statistics}, 33\penalty0 (3):\penalty0
  1138--1154, 06 2005.

\bibitem[Leiner et~al.(2025)Leiner, Duan, Wasserman, and
  Ramdas]{leiner2025data}
James Leiner, Boyan Duan, Larry Wasserman, and Aaditya Ramdas.
\newblock Data fission: splitting a single data point.
\newblock \emph{Journal of the American Statistical Association}, 120\penalty0
  (549):\penalty0 135--146, 2025.

\bibitem[Liang et~al.(2024)Liang, Xiang, Mei, and Li]{liang2024directional}
Wenjuan Liang, Dongdong Xiang, Yajun Mei, and Wendong Li.
\newblock Directional false discovery rate control in large-scale multiple
  comparisons.
\newblock \emph{Journal of Applied Statistics}, 51\penalty0 (15):\penalty0
  3195--3214, 2024.

\bibitem[Luo et~al.(2023)Luo, Ebadi, Emery, He, Noble, and
  Keich]{luo:competition}
Dong Luo, Arya Ebadi, Kristen Emery, Yilun He, William~Stafford Noble, and Uri
  Keich.
\newblock Competition-based control of the false discovery proportion.
\newblock \emph{Biometrics}, 79\penalty0 (4):\penalty0 3472--3484, 2023.

\bibitem[Pournaderi and Xiang(2024)]{pournaderi2024variable}
Mehrdad Pournaderi and Yu~Xiang.
\newblock Variable selection with the knockoffs: Composite null hypotheses.
\newblock \emph{Journal of Statistical Planning and Inference}, 231:\penalty0
  106119, 2024.
\newblock \doi{10.1016/j.jspi.2023.106119}.

\bibitem[Rajchert and Keich(2023)]{rajchert:controlling}
Andrew Rajchert and Uri Keich.
\newblock Controlling the false discovery rate via competition: Is the +1
  needed?
\newblock \emph{Statistics \& Probability Letters}, 197:\penalty0 109819, 2023.
\newblock ISSN 0167-7152.
\newblock \doi{https://doi.org/10.1016/j.spl.2023.109819}.

\bibitem[Ren and Barber(2024)]{ren:derandomised}
Zhimei Ren and Rina~Foygel Barber.
\newblock Derandomised knockoffs: leveraging e-values for false discovery rate
  control.
\newblock \emph{Journal of the Royal Statistical Society Series B}, 86\penalty0
  (1):\penalty0 122--154, 2024.

\bibitem[Rhee et~al.(2006)Rhee, Taylor, Wadhera, Ben-Hur, Brutlag, and
  Shafer]{rhee:genotypic}
Soo-Yon Rhee, Jonathan Taylor, Gauhar Wadhera, Asa Ben-Hur, Douglas~L Brutlag,
  and Robert~W Shafer.
\newblock Genotypic predictors of human immunodeficiency virus type 1 drug
  resistance.
\newblock \emph{Proceedings of the National Academy of Sciences}, 103\penalty0
  (46):\penalty0 17355--17360, 2006.

\bibitem[Spector and Janson(2022)]{spector2022powerful}
Asher Spector and Lucas Janson.
\newblock Powerful knockoffs via minimizing reconstructability.
\newblock \emph{The Annals of Statistics}, 50\penalty0 (1):\penalty0 252--276,
  2022.

\bibitem[Tibshirani(1996)]{tibshirani:regression}
R.~J. Tibshirani.
\newblock Regression shrinkage and selection via the lasso.
\newblock \emph{Journal of the Royal Statistical Society B}, 58\penalty0
  (1):\penalty0 267--288, 1996.

\bibitem[van~der Laan et~al.(2004)van~der Laan, Dudoit, and
  Pollard]{vanderLaan:augmentation}
Mark~J. van~der Laan, Sandrine Dudoit, and Katherine~S. Pollard.
\newblock Augmentation procedures for control of the generalized family-wise
  error rate and tail probabilities for the proportion of false positives.
\newblock \emph{Statistical Applications in Genetics and Molecular Biology},
  3\penalty0 (1), 2004.

\bibitem[Xing et~al.(2023)Xing, Zhao, and Liu]{xing2023controlling}
Xin Xing, Zhigen Zhao, and Jun~S Liu.
\newblock Controlling false discovery rate using gaussian mirrors.
\newblock \emph{Journal of the American Statistical Association}, 118\penalty0
  (541):\penalty0 222--241, 2023.

\bibitem[Zou and Hastie(2005)]{zou:regularization}
Hui Zou and Trevor Hastie.
\newblock Regularization and variable selection via the elastic net.
\newblock \emph{Journal of the Royal Statistical Society Series B}, 67\penalty0
  (2):\penalty0 301--320, 2005.

\end{thebibliography}

\clearpage

\appendix
\renewcommand{\thesection}{A}
\renewcommand{\thealgorithm}{A\arabic{algorithm}}
\renewcommand{\thetable}{A\arabic{table}}
\renewcommand{\thefigure}{A\arabic{figure}}
\renewcommand{\theLemma}{A\arabic{Lemma}}
\renewcommand{\theDefinition}{A\arabic{Definition}}
\renewcommand{\theRemark}{A\arabic{Remark}}
\setcounter{algorithm}{0}

\section{Algorithms}
\label{appendix:algorithms}

\begin{algorithm}[H]
	\caption{ {\bf  FDP-SD} \citep{luo:competition}}
	\label{alg:FDP-SD}
	\begin{algorithmic}[1]
		\Require  \begin{tabular}[t]{p{0.8\textwidth}}
		$\mathbf W \equiv (W_i)_{i = 1}^m$ the list of statistics; $\gamma \in (0, 1)$ the FDX threshold; $\alpha \in (0, 1)$ the FDP threshold
		\end{tabular}
		\Ensure A discovery list $\mathcal{\hat S}_{\text{FDP-SD}}$
			\State $i_0 \gets \max \{1, \lceil (\lceil \log_{2}(1/\gamma) \rceil - 1 )/\alpha \rceil \}$
            \If{$i_0>m$}
                \State \Return $\emptyset$
            \EndIf
            \State Let $\pi$ be a permutation such that $\lvert W_{\pi(1)} \rvert \geq \lvert W_{\pi(2)} \rvert \geq \dots \geq \lvert W_{\pi(m)} \rvert$
            \For{$i = i_0, \dots, m$}
            \State $D_i \gets \#\{j\le i : \sign(W_{\pi(j)}) = -1 \}$ 
            \State $\delta_i \gets \max \{ d \in \{0, 1, \dots, i \} : F_{B(\lfloor (i-d)\alpha\rfloor+1+d,\,1/2)}(d) \leq \gamma \}$ \Comment{$F_{B(n, p)}$ denotes the CDF of a Binomial$(n, p)$ RV}
            \EndFor
            \If{$D_{i_0} \leq \delta_{i_0}$}
                \State $\tau \gets \max \{i \in \{i_0, \dots, m \} : D_j \leq \delta_j, \forall j = i_0, \dots, i \}$
                \State \Return $\mathcal{\hat S}_{\text{FDP-SD}} \gets \{\pi(j): j\in[m],\ j\leq\tau,\ W_{\pi(j)}>0\}$
            \Else
                \State \Return $\mathcal{\hat S}_{\text{FDP-SD}} \gets \emptyset$
            \EndIf
		\end{algorithmic}
\end{algorithm}

\newpage

\begin{algorithm}[H]
	\caption{ {\bf FDX-response-guided knockoff filter}}
	\label{alg:response-kf-fdx}
	\begin{algorithmic}[1]
		\Require  \begin{tabular}[t]{p{0.8\textwidth}}
		$\mathbf y \in \mathbb R^n$ the response vector; $\mathbf X \in \mathbb R^{n \times p}$ the design matrix with $n > p + 2$;  $d \in \{\pm 1\}$ the prespecified target sign of interest; $\gamma \in (0, 1)$ the FDX threshold; $\alpha \in (0, 1)$ the FDP threshold
		\end{tabular}
		\Ensure A discovery list $\mathcal{\hat S}_{\text{FDX-RGKF}}$
        \State $\tilde{\mathbf y} \gets \mathbf y + \tilde{\boldsymbol{\varepsilon}}$ where $\tilde{\boldsymbol{\varepsilon}} \mid \mathbf y \sim \mathcal N(\mathbf 0, s_y^2 \mathbf I_n)$ \Comment{$s_y^2$ is the sample variance of $\mathbf y$}
        \State $\mathbf s \gets s(\mathbf X, \tilde{\mathbf y}, d) \in [0,1]^p$ \Comment{$\mathbf s$ satisfies Condition~\ref{ass:correct}}
        \State $\tilde{\mathbf X} \gets g(\mathbf X, \tilde{\mathbf y}, \mathbf s)$ \Comment{$\tilde{\mathbf X}$ satisfies $\mathbf G^{(e)} = {\hat{\mathbf X}}^{(e)\top} \hat{\mathbf X}^{(e)}$}
        \State $\hat{\mathbf X} \gets [\mathbf X,  \tilde{\mathbf X}]$
        \State $\mathbf W \gets \mathbf W(\hat{\mathbf X}, \mathbf y)$ \Comment{$\mathbf W$ satisfies Conditions~\ref{ass:suff} and \ref{ass:sym}}
        \State $S_j \gets \sign\bigl((\mathbf X_{j}-\widetilde{\mathbf X}_{j})^{\!\top}\mathbf y\bigr) \cdot \sign\left(W_j \right)$ for $j = 1, \ldots, p$
        \State $\mathcal J_d \gets \{ j \in [p] : S_j = d  \}$
        \State \Return $\mathcal{\hat S}_{\text{FDX-RGKF}} \gets \text{FDP-SD}((W_j)_{j \in \mathcal J_d}, \gamma, \alpha)$ \Comment{Algorithm~\ref{alg:FDP-SD}}
		\end{algorithmic}
\end{algorithm}

\begin{algorithm}[H]
	\caption{ {\bf JS} \citep{janson:familywise}}
	\label{alg:JS}
	\begin{algorithmic}[1]
		\Require  \begin{tabular}[t]{p{0.8\textwidth}}
		$\mathbf W \equiv (W_i)_{i = 1}^m$ the list of statistics; $k \in \mathbb{N}_{>0}$ for $k$-FWER; $\alpha \in (0, 1)$ the $k$-FWER threshold\\
		\end{tabular}
		\Ensure A discovery list $\mathcal{\hat S}_{\text{JS}}$
            \State $v \gets \max \left(\{0\}\cup\left\{w \in \mathbb{N}_{>0} : \sum_{i=k}^{\infty}2^{-i-w}\binom{i+w-1}{i}\leq\alpha\right\}\right)$
            \If{$v=0$}
                \State \Return $\emptyset$
            \EndIf
            \State Let $\pi$ order the scores by decreasing absolute value
            \State $\tau \gets \min\left(\left\{r\in[m]:\#\{j\leq r:W_{\pi(j)}<0\}=v\right\}\cup\{m+1\}\right)-1$
            \State \Return $\mathcal{\hat S}_{\text{JS}} \gets \{\pi(j):j\in[m],\ j\leq\tau,\ W_{\pi(j)}>0\}$
		\end{algorithmic}
\end{algorithm}

\begin{algorithm}[H]
	\caption{ {\bf $k$-FWER-response-guided knockoff filter}}
	\label{alg:response-kf-kfwer}
	\begin{algorithmic}[1]
		\Require  \begin{tabular}[t]{p{0.8\textwidth}}
		$\mathbf y \in \mathbb R^n$ the response vector; $\mathbf X \in \mathbb R^{n \times p}$ the design matrix with $n > p + 2$;  $d \in \{\pm 1\}$ the prespecified target sign of interest; $k \in \mathbb{N}_{>0}$ for $k$-FWER; $\alpha \in (0, 1)$ the $k$-FWER threshold\\
		\end{tabular}
		\Ensure A discovery list $\mathcal{\hat S}_{\text{k-FWER-RGKF}}$
        \State $\tilde{\mathbf y} \gets \mathbf y + \tilde{\boldsymbol{\varepsilon}}$ where $\tilde{\boldsymbol{\varepsilon}} \mid \mathbf y \sim \mathcal N(\mathbf 0, s_y^2 \mathbf I_n)$ \Comment{$s_y^2$ is the sample variance of $\mathbf y$}
        \State $\mathbf s \gets s(\mathbf X, \tilde{\mathbf y}, d) \in [0,1]^p$ \Comment{$\mathbf s$ satisfies Condition~\ref{ass:correct}}
        \State $\tilde{\mathbf X} \gets g(\mathbf X, \tilde{\mathbf y}, \mathbf s)$ \Comment{$\tilde{\mathbf X}$ satisfies $\mathbf G^{(e)} = {\hat{\mathbf X}}^{(e)\top} \hat{\mathbf X}^{(e)}$}
        \State $\hat{\mathbf X} \gets [\mathbf X,  \tilde{\mathbf X}]$
        \State $\mathbf W \gets \mathbf W(\hat{\mathbf X}, \mathbf y)$ \Comment{$\mathbf W$ satisfies Conditions~\ref{ass:suff} and \ref{ass:sym}}
        \State $S_j \gets \sign\bigl((\mathbf X_{j}-\widetilde{\mathbf X}_{j})^{\!\top}\mathbf y\bigr) \cdot \sign\left(W_j \right)$ for $j = 1, \ldots, p$
        \State $\mathcal J_d \gets \{ j \in [p] : S_j = d  \}$
        \State \Return $\mathcal{\hat S}_{\text{k-FWER-RGKF}} \gets \text{JS}((W_j)_{j \in \mathcal J_d}, k, \alpha)$ \Comment{Algorithm~\ref{alg:JS}}
		\end{algorithmic}
\end{algorithm}

\clearpage
\appendix
\renewcommand{\thesection}{B}
\renewcommand{\thealgorithm}{B\arabic{algorithm}}
\renewcommand{\thetable}{B\arabic{table}}
\renewcommand{\thefigure}{B\arabic{figure}}
\renewcommand{\theLemma}{B\arabic{Lemma}}
\renewcommand{\theTheorem}{B\arabic{Theorem}}
\renewcommand{\theCondition}{B\arabic{Condition}}
\renewcommand{\theRemark}{B\arabic{Remark}}
\renewcommand{\theDefinition}{B\arabic{Definition}}
\setcounter{figure}{0}
\setcounter{Lemma}{0}
\setcounter{Remark}{0}
\setcounter{Definition}{0}
\setcounter{Theorem}{0}

\section{Proofs}
\label{appendix:proofs}

\subsection{Proof of Proposition~\ref{lem:signed-kf}}

\begin{proof}
    From Lemma 1 in the supplementary material of \cite{barber2019knockoff}, the $\sigma$-algebra $\mathcal{V}$ generated by $\{(\mathbf X + \tilde{\mathbf X})^{\!\top} \mathbf y, \lvert (\mathbf X - \tilde{\mathbf X})^{\!\top} \mathbf y \rvert \}$ contains $\mathbf S = \left(S_i\right)_{i = 1}^p$ and $\lvert \mathbf W \rvert$ where $S_j \equiv \sign\bigl((\mathbf X_{j}-\widetilde{\mathbf X}_{j})^{\!\top}\mathbf y\bigr) \cdot \sign\left (W_j \right)$.
    Therefore, conditioned on $\mathcal V$, the feature indices $\mathcal J \equiv \{j \in [p] : S_j= d \}$ are fixed.
    It also follows from the same lemma that conditioned on $\mathcal V$, the signs of $(W_j)_{j \in \mathcal J}$ are mutually independent with $\mathbb{P}(W_j < 0 \mid \mathcal V) \geq 1/2$ for $j \in \mathcal J \cap \mathcal{H}_0^d$.
    Therefore, Corollary 1 in \cite{barber2019knockoff} and the argument in their proof of Theorem 4 imply that applying SSS+ within $\mathcal J$ controls the FDR with respect to $\mathcal H_0^d$. Since all selected features have estimated sign $d$, this is the prespecified directional FDR $\text{FDR}_{\text{dir}}^{d}$ defined in (\ref{eq:fdr_dir_d}).
\end{proof}

\subsection{Key Lemmas for Theorems~\ref{theorem:fdr_control} and~\ref{theorem:fdx_fwer_control}}
Let $\mathbf y \sim \mathcal N (\mathbf X \boldsymbol\beta, \sigma^2 \mathbf I_n)$ with $\sigma^2>0$. Denote $\tilde{\mathbf y}$ and $\tilde{\mathbf X}$ as in the response-guided knockoff filter.
Conditional on $\tilde{\mathbf y}$, the separation parameters, the knockoff matrix, and its Gram matrix are fixed.
Let $\mathcal A \equiv \{j\in[p]:s_j>0\}$ be the set of active coordinates. This set is determined by $\tilde{\mathbf y}$. The following construction is made separately for every $j\in\mathcal A$. Write
\[
    \mathbf a_j \equiv \mathbf X_j-\tilde{\mathbf X}_j \neq \mathbf 0.
\]
Conditional on $\tilde{\mathbf y}$, use a fixed measurable completion rule to choose vectors $\{\mathbf b_\ell^{(j)}:\ell\in[n]\setminus\{j\}\}$ forming an orthonormal basis of $\mathbf a_j^\perp$, so that the following matrix is a measurable function of $\tilde{\mathbf y}$:
\begin{align}
    \mathbf A^{(j)}
    &\equiv
    \bigl[
        \mathbf b_1^{(j)},\ldots,\mathbf b_{j-1}^{(j)},
        \mathbf a_j,
        \mathbf b_{j+1}^{(j)},\ldots,\mathbf b_n^{(j)}
    \bigr]^\top \in \mathbb R^{n\times n},
    &\mathbf z^{(j)} &\equiv \mathbf A^{(j)}\mathbf y.
    \label{eq:coordinate_transform}
\end{align}
For every $j\in\mathcal A$, the matrix $\mathbf A^{(j)}$ is full rank, its $j$th row is $\mathbf a_j^\top$, and hence $z_j^{(j)}=(\mathbf X_j-\tilde{\mathbf X}_j)^\top\mathbf y$.
The knockoff Gram identities from Condition~\ref{ass:correct} imply
\begin{align}
    \mathbf a_j^\top(\mathbf X_r-\tilde{\mathbf X}_r)&=0 &&\text{for every }j\in\mathcal A\text{ and }r\neq j,\notag\\
    \mathbf a_j^\top(\mathbf X_r+\tilde{\mathbf X}_r)&=0 &&\text{for every }j\in\mathcal A\text{ and }r\in[p].
    \label{eq:coordinate_orthogonality}
\end{align}
For any vector $\mathbf x$, let $\mathbf x_{-\ell}$ denote $\mathbf x$ with its $\ell$th coordinate removed, and let $\mathbf x_{\ell\gets w}$ denote $\mathbf x$ with its $\ell$th coordinate replaced by $w$.
Consequently, for each $j\in\mathcal A$, all other feature--knockoff contrasts and all feature--knockoff sums lie in $\mathbf a_j^\perp$, and their inner products with $\mathbf y$ are functions of $\mathbf z_{-j}^{(j)}$.
Define the common conditioning field
\begin{equation}
    \mathcal V
    \equiv \sigma\!\left(
        \tilde{\mathbf y},
        \lvert(\mathbf X-\tilde{\mathbf X})^\top\mathbf y\rvert,
        (\mathbf X+\tilde{\mathbf X})^\top\mathbf y
    \right).
    \label{eq:rgkf_conditioning_field}
\end{equation}

\begin{Lemma}[\cite{barber2019knockoff}]
    \label{lemma:barber}
    For any $j\in\mathcal A$, let
    \[
        \mathcal H_j \equiv \sigma\!\left(\mathbf z_{-j}^{(j)},\lvert z_j^{(j)}\rvert,\tilde{\mathbf y}\right).
    \]
    Then $\lvert \mathbf W \rvert$ and $\mathbf S$ are $\mathcal V$-measurable, and $\mathcal V\subseteq\mathcal H_j$. In particular, both are $\mathcal H_j$-measurable.
\end{Lemma}
\begin{proof}
    Conditional on $\tilde{\mathbf y}$, the knockoff Gram matrix is fixed. Since the response-guided knockoff statistics satisfy Conditions~\ref{ass:suff} and~\ref{ass:sym}, Lemma~1 in \cite{barber2019knockoff} gives
    \begin{itemize}
        \item $\lvert \mathbf W \rvert$ is $\mathcal V$-measurable, and
        \item $\mathbf S$ is $\mathcal V$-measurable.
    \end{itemize}
    By \eqref{eq:coordinate_orthogonality}, $\mathcal V\subseteq\mathcal H_j$, which proves the result.
\end{proof}

The following conditional sign-odds identity is standard and likely appears elsewhere; we nevertheless include a proof for completeness.

\begin{Lemma}
\label{lemma:sign-odds}
Let $\mathbf Z\in\mathbb R^n$ be a random vector and fix $j\in[n]$.
Write
\[
\mathbf U:=\mathbf Z_{-j}\in\mathbb R^{n-1},\qquad R:=|Z_j|\in[0,\infty).
\]
Assume that $\mathbf Z$ has a Borel-measurable density $f_{\mathbf Z}:\mathbb R^n\to[0,\infty)$. Let $\mathcal G\equiv\sigma(\mathbf U,R)$, and let $T\in\{-1,+1\}$ be $\mathcal G$-measurable. Then the following holds almost surely:
\begin{align}
\mathbb{P}\bigl(\sign(Z_j)=T\mid \mathcal G\bigr)
&=\frac{f_{\mathbf Z}\bigl(\mathbf Z_{j\gets TR}\bigr)}
{f_{\mathbf Z}\bigl(\mathbf Z_{j\gets TR}\bigr)+f_{\mathbf Z}\bigl(\mathbf Z_{j\gets -TR}\bigr)},\label{eq:condplus}\\[0.5em]
\mathbb{P}\bigl(\sign(Z_j)=-T\mid \mathcal G\bigr)
&=\frac{f_{\mathbf Z}\bigl(\mathbf Z_{j\gets -TR}\bigr)}
{f_{\mathbf Z}\bigl(\mathbf Z_{j\gets TR}\bigr)+f_{\mathbf Z}\bigl(\mathbf Z_{j\gets -TR}\bigr)}.\label{eq:condminus}
\end{align}
\end{Lemma}

\begin{proof}
First consider the first identity for $T\equiv+1$. Since $\mathbf Z$ has a density, $\mathbb P(Z_j=0)=0$.
For $\mathbf z\in\mathbb R^n$, define
\[
\psi(\mathbf z)
\equiv
\frac{f_{\mathbf Z}\bigl(\mathbf z_{j\gets\lvert z_j\rvert}\bigr)}
{f_{\mathbf Z}\bigl(\mathbf z_{j\gets\lvert z_j\rvert}\bigr)+f_{\mathbf Z}\bigl(\mathbf z_{j\gets-\lvert z_j\rvert}\bigr)}
\]
when the denominator is positive, and set $\psi(\mathbf z)=0$ otherwise. Since $f_{\mathbf Z}$ is Borel measurable and $\psi(\mathbf z)$ depends only on $(\mathbf z_{-j},\lvert z_j\rvert)$, $\psi(\mathbf Z)$ is $\mathcal G$-measurable. Fix an arbitrary Borel set $B\subseteq\mathbb R^{n-1}\times[0,\infty)$.
Then by definition of conditional expectation, what remains to show is that
\[ 
\mathbb E \left(  1_{\{Z_j>0\}} \cdot  1_{\{(\mathbf U,R) \in B\}} \right) = \mathbb E \left( \psi(\mathbf Z) \cdot 1_{\{(\mathbf U,R) \in B\}} \right).
\]
The left-hand side equals
\[
\int_{\{\mathbf z:z_j>0\}} 1_{\{(\mathbf z_{-j},z_j)\in B\}} \cdot f_{\mathbf Z}(\mathbf z)\,d\mathbf z.
\]
The right-hand side equals
\begin{align*}
&\int_{\mathbb R^n}1_{\{(\mathbf z_{-j},\lvert z_j\rvert)\in B\}}\cdot\psi(\mathbf z)\cdot f_{\mathbf Z}(\mathbf z)\,d\mathbf z\\
&\quad=\int_{\{\mathbf z:z_j>0\}} 1_{\{(\mathbf z_{-j},z_j)\in B\}}\cdot\psi(\mathbf z)
\left[f_{\mathbf Z}(\mathbf z)+f_{\mathbf Z}\bigl(\mathbf z_{j\gets-z_j}\bigr)\right] \,d\mathbf z\\
&\quad=\int_{\{\mathbf z:z_j>0\}} 1_{\{(\mathbf z_{-j},z_j)\in B\}}\cdot f_{\mathbf Z}(\mathbf z)\,d\mathbf z,
\end{align*}
which is the same expression as the left-hand side above.
The corresponding identity for a negative sign follows by complementation. Both displayed identities for a $\mathcal G$-measurable $T$ then follow by partitioning over the events $\{T=+1\}$ and $\{T=-1\}$.
\end{proof}

\begin{Lemma}
    \label{lemma:rgkf-density-ratio}
    For any $j\in\mathcal A$, let $\mathbf z^{(j)}$ be defined by \eqref{eq:coordinate_transform} and set $r=\lvert z_j^{(j)}\rvert$. Conditional on $\tilde{\mathbf y}$, let $f_j(\cdot\mid\tilde{\mathbf y})$ denote the density of $\mathbf z^{(j)}$. Then, almost surely,
    \begin{equation}
    \frac{ f_j\left( \mathbf z^{(j)}_{j \gets r} \mid \tilde{\mathbf y}\right)}
    {f_j\left( \mathbf z^{(j)}_{j \gets -r} \mid \tilde{\mathbf y}\right)}
    =\exp\!\left(\frac{\beta_jr}{\sigma^2}\right).
    \label{eq:rgkf_exact_density_ratio}
    \end{equation}
    Moreover,
    \begin{align}
        \label{eq:rgkf_cond_density_z}
        f_j(\mathbf c\mid\tilde{\mathbf y})
        \propto g^{(j)}(\mathbf c)\,
        f_{\tilde{\boldsymbol\varepsilon}\mid\mathbf y}
        \left(
            \tilde{\mathbf y}-(\mathbf A^{(j)})^{-1}\mathbf c
            \mathrel{\Big|}(\mathbf A^{(j)})^{-1}\mathbf c
        \right),
    \end{align}
    where $g^{(j)}$ is the $\mathcal N\!\left(\mathbf A^{(j)}\mathbf X\boldsymbol\beta,\sigma^2\mathbf A^{(j)}\mathbf A^{(j)\top}\right)$ density and the proportionality constant does not depend on $\mathbf c$.
\end{Lemma}
\begin{proof}
    Fix a realisation of $\tilde{\mathbf y}$. The conditional density of $\mathbf y$ is
    \begin{align*}
        f_{\mathbf y\mid\tilde{\mathbf y}}(\mathbf a\mid\tilde{\mathbf y})
        &=\frac{f_{\mathbf y}(\mathbf a)\,
        f_{\tilde{\boldsymbol\varepsilon}\mid\mathbf y}(\tilde{\mathbf y}-\mathbf a\mid\mathbf a)}
        {f_{\tilde{\mathbf y}}(\tilde{\mathbf y})}.
    \end{align*}
    Formula~\eqref{eq:rgkf_cond_density_z} follows by applying the change of variables $\mathbf c=\mathbf A^{(j)}\mathbf y$ to the conditional density of $\mathbf y$ given $\tilde{\mathbf y}$, where the Gaussian density $g^{(j)}$ includes the Jacobian $\lvert\det(\mathbf A^{(j)})^{-1}\rvert$.

    To compare the two sides of the ratio in the lemma, define
    \begin{align*}
        \mathbf c_+&\equiv\mathbf z^{(j)}_{j\gets r},
        &\mathbf c_-&\equiv\mathbf z^{(j)}_{j\gets-r},\\
        \mathbf x_+&\equiv(\mathbf A^{(j)})^{-1}\mathbf c_+,
        &\mathbf x_-&\equiv(\mathbf A^{(j)})^{-1}\mathbf c_-.
    \end{align*}
    Taking the ratio of \eqref{eq:rgkf_cond_density_z} at $\mathbf c_+$ and $\mathbf c_-$ gives
    \begin{align}
        \frac{f_j(\mathbf c_+\mid\tilde{\mathbf y})}
        {f_j(\mathbf c_-\mid\tilde{\mathbf y})}
        &={\frac{g^{(j)}(\mathbf c_+)}{g^{(j)}(\mathbf c_-)}}
        \cdot
        \frac{
            f_{\tilde{\boldsymbol\varepsilon}\mid\mathbf y}
            (\tilde{\mathbf y}-\mathbf x_+\mid\mathbf x_+)
        }{
            f_{\tilde{\boldsymbol\varepsilon}\mid\mathbf y}
            (\tilde{\mathbf y}-\mathbf x_-\mid\mathbf x_-)
        }.
        \label{eq:factor_density_ratio}
    \end{align}
    We evaluate the two factors separately.

    First consider the Gaussian factor. Under $g^{(j)}$, the covariance between the $j$th coordinate and any coordinate $\ell\neq j$ is
    \[
        \operatorname{Cov}_{g^{(j)}}(z_j^{(j)},z_\ell^{(j)})
        =\sigma^2\mathbf a_j^\top\mathbf b_\ell^{(j)}
        =0,
    \]
    because $\mathbf b_\ell^{(j)}\in\mathbf a_j^\perp$. Since $g^{(j)}$ is a multivariate Gaussian density, zero covariance implies that $z_j^{(j)}$ is independent of $\mathbf z_{-j}^{(j)}$. The vectors $\mathbf c_+$ and $\mathbf c_-$ agree outside their $j$th coordinate, so
    \[
        \frac{g^{(j)}(\mathbf c_+)}{g^{(j)}(\mathbf c_-)}
        =\frac{g_j^{(j)}(r)}{g_j^{(j)}(-r)},
    \]
    where $g_j^{(j)}$ is the marginal density of
    \[
        z_j^{(j)}\sim\mathcal N\!\left(
            \mathbf a_j^\top\mathbf X\boldsymbol\beta,
            \sigma^2\lVert\mathbf a_j\rVert^2
        \right)
    \]
    under $g^{(j)}$. Condition~\ref{ass:gram} gives $\mathbf a_j^\top\mathbf X\boldsymbol\beta=s_j\beta_j$. For a $\mathcal N(\theta,\tau^2)$ density $h$, we have $h(t)/h(-t)=\exp(2\theta t/\tau^2)$. Therefore,
    \begin{align}
        \frac{g^{(j)}(\mathbf c_+)}{g^{(j)}(\mathbf c_-)}
        =\exp\!\left(
            \frac{2s_j\beta_jr}{\sigma^2\lVert\mathbf a_j\rVert^2}
        \right).
        \label{eq:gaussian_density_ratio}
    \end{align}
    Since $\lVert\mathbf a_j\rVert^2=2s_j$, the Gaussian factor in \eqref{eq:gaussian_density_ratio} equals $\exp(\beta_jr/\sigma^2)$.

    We next show that the perturbation-density factor in \eqref{eq:factor_density_ratio} equals one. Since
    $\tilde{\boldsymbol\varepsilon}\mid\mathbf y=\mathbf x\sim
    \mathcal N(\mathbf0,s_{\mathbf x}^2\mathbf I_n)$, its density at $\tilde{\mathbf y}-\mathbf x$ is
    \[
        f_{\tilde{\boldsymbol\varepsilon}\mid\mathbf y}
        (\tilde{\mathbf y}-\mathbf x\mid\mathbf x)
        =\left(2\pi s_{\mathbf x}^2\right)^{-n/2}
        \exp\!\left(
            -\frac{\lVert\tilde{\mathbf y}-\mathbf x\rVert^2}
            {2s_{\mathbf x}^2}
        \right).
    \]
    It therefore suffices to prove
    \begin{equation}
        \lVert\tilde{\mathbf y}-\mathbf x_-\rVert
        =\lVert\tilde{\mathbf y}-\mathbf x_+\rVert
        \quad\text{and}\quad
        s_{\mathbf x_-}^2=s_{\mathbf x_+}^2.
        \label{eq:reflection_invariants}
    \end{equation}

    Define the Householder reflection
    \[
        \mathbf H_j
        \equiv\mathbf I_n-2\frac{\mathbf a_j\mathbf a_j^\top}
        {\lVert\mathbf a_j\rVert^2}.
    \]
    For $\ell\neq j$, we have $\mathbf b_\ell^{(j)\top}\mathbf H_j=\mathbf b_\ell^{(j)\top}$, while $\mathbf a_j^\top\mathbf H_j=-\mathbf a_j^\top$. Therefore, for every $\ell\neq j$,
    \[
        (\mathbf A^{(j)}\mathbf H_j\mathbf x_+)_\ell
        =\mathbf b_\ell^{(j)\top}\mathbf x_+
        =(\mathbf c_+)_\ell
        =(\mathbf c_-)_\ell,
    \]
    whereas the $j$th coordinate satisfies
    \[
        (\mathbf A^{(j)}\mathbf H_j\mathbf x_+)_j
        =-\mathbf a_j^\top\mathbf x_+
        =-r
        =(\mathbf c_-)_j.
    \]
    Thus $\mathbf A^{(j)}\mathbf H_j\mathbf x_+=\mathbf c_-=\mathbf A^{(j)}\mathbf x_-$, and the invertibility of $\mathbf A^{(j)}$ gives $\mathbf H_j\mathbf x_+=\mathbf x_-$.

    Condition~\ref{ass:correct} gives $\mathbf a_j^\top\tilde{\mathbf y}=0$ and $\mathbf a_j^\top\mathbf1=0$, so
    \[
        \mathbf H_j\tilde{\mathbf y}=\tilde{\mathbf y}
        \qquad\text{and}\qquad
        \mathbf H_j\mathbf1=\mathbf1.
    \]
    Because $\mathbf H_j$ is orthogonal,
    \begin{align*}
        \lVert\tilde{\mathbf y}-\mathbf x_-\rVert
        &=\lVert\mathbf H_j\tilde{\mathbf y}-\mathbf H_j\mathbf x_+\rVert
        =\lVert\mathbf H_j(\tilde{\mathbf y}-\mathbf x_+)\rVert
        =\lVert\tilde{\mathbf y}-\mathbf x_+\rVert.
    \end{align*}
    Let $\bar x_\pm\equiv n^{-1}\mathbf1^\top\mathbf x_\pm$ denote the sample means of $\mathbf x_\pm$. Since $\mathbf x_-=\mathbf H_j\mathbf x_+$, $\mathbf H_j$ is symmetric, and $\mathbf H_j\mathbf1=\mathbf1$, we have
    \[
        \bar x_-
        =\frac1n\mathbf1^\top\mathbf H_j\mathbf x_+
        =\frac1n(\mathbf H_j\mathbf1)^\top\mathbf x_+
        =\bar x_+,
    \]
    and consequently
    \[
        \mathbf x_- -\bar x_-\mathbf1
        =\mathbf H_j\bigl(\mathbf x_+-\bar x_+\mathbf1\bigr).
    \]
    Since sample variance is a fixed multiple of the squared norm of the centered vector, orthogonality of $\mathbf H_j$ now gives $s_{\mathbf x_-}^2=s_{\mathbf x_+}^2$. This proves \eqref{eq:reflection_invariants} and therefore,
    \[
        \frac{
            f_{\tilde{\boldsymbol\varepsilon}\mid\mathbf y}
            (\tilde{\mathbf y}-\mathbf x_+\mid\mathbf x_+)
        }{
            f_{\tilde{\boldsymbol\varepsilon}\mid\mathbf y}
            (\tilde{\mathbf y}-\mathbf x_-\mid\mathbf x_-)
        }
        =1.
    \]

    Combining this identity with \eqref{eq:factor_density_ratio} and \eqref{eq:gaussian_density_ratio} proves the result.
\end{proof}

\begin{Lemma}[Conditional independence of the signs]
    \label{lemma:rgkf_sign_independence}
    For $j\in\mathcal A$, write $r_j\equiv|z_j^{(j)}|$.
    Conditional on $\mathcal V$ in~\eqref{eq:rgkf_conditioning_field}, the signs $E_j\equiv\sign(z_j^{(j)})$, $j\in\mathcal A$, are mutually independent, with
    \begin{equation}
        p_j\equiv\mathbb P(E_j=1\mid\mathcal V)
        =\frac{\exp(\beta_jr_j/\sigma^2)}{1+\exp(\beta_jr_j/\sigma^2)}.
        \label{eq:rgkf_raw_sign_probability}
    \end{equation}
    Consequently, the nonzero signs of $\mathbf W$ are mutually independent conditional on $\mathcal V$, and, on $\{|W_j|>0\}$,
    \begin{equation}
        \mathbb P(W_j>0\mid\mathcal V)
        =\frac{\exp(S_j\beta_jr_j/\sigma^2)}{1+\exp(S_j\beta_jr_j/\sigma^2)}.
        \label{eq:rgkf_score_sign_probability}
    \end{equation}
\end{Lemma}

\begin{proof}
    For $j\in\mathcal A$, $z_j^{(j)}$ is nonzero almost surely, since $\mathbf z^{(j)}$ has a conditional density. Apply Lemma~\ref{lemma:sign-odds} with the fixed target sign $+1$ and use Lemma~\ref{lemma:rgkf-density-ratio} to obtain
    \[
        \mathbb P(E_j=1\mid\mathcal H_j)
        =\frac{\exp(\beta_jr_j/\sigma^2)}{1+\exp(\beta_jr_j/\sigma^2)}
        =p_j.
    \]
    By Lemma~\ref{lemma:barber} and~\eqref{eq:coordinate_orthogonality}, $\mathcal V$ and every $E_k$ with $k\in\mathcal A\setminus\{j\}$ are contained in $\mathcal H_j$. Since $p_j$ only depends on $r_j$ it is $\mathcal V$-measurable, and the tower property gives
    \[
        \mathbb P\bigl(E_j=1\mid\mathcal V,(E_k)_{k\in\mathcal A\setminus\{j\}}\bigr)=p_j.
    \]
    Conditional on $\mathcal V$, the active set $\mathcal A$ is fixed. Since the preceding equality holds for every $j\in\mathcal A$ with $\mathcal V$-measurable $p_j$, it establishes mutual conditional independence of the signs, with $\mathbb P(E_j=1\mid\mathcal V)=p_j$.

    If $s_j=0$, then $\mathbf X_j=\tilde{\mathbf X}_j$, so antisymmetry forces $W_j=0$. Thus every index with a nonzero score belongs to $\mathcal A$.
    By Lemma~\ref{lemma:barber}, $|\mathbf W|$ and $\mathbf S$ are $\mathcal V$-measurable. Hence the indices of nonzero scores and their corresponding $S_j$ values are fixed given $\mathcal V$. On these indices, $S_j=E_j\sign(W_j)\in\{-1,+1\}$, so $\sign(W_j)=S_jE_j$. Conditional on $\mathcal V$, multiplying the mutually independent signs $E_j$ by the fixed signs $S_j$ preserves mutual independence. Therefore, the nonzero signs of $\mathbf W$ are mutually independent conditional on $\mathcal V$.

    Moreover, on $\{|W_j|>0\}$,
    \begin{align*}
        \mathbb P(W_j>0\mid\mathcal V)
        &=\mathbb P(E_j=S_j\mid\mathcal V)\\
        &=\begin{cases}
            p_j, & S_j=+1,\\
            1-p_j, & S_j=-1,
        \end{cases}\\
        &=\frac{\exp(S_j\beta_jr_j/\sigma^2)}
        {1+\exp(S_j\beta_jr_j/\sigma^2)},
    \end{align*}
    where the last equality follows from~\eqref{eq:rgkf_raw_sign_probability}. This proves~\eqref{eq:rgkf_score_sign_probability}.
\end{proof}

In Assumption~\ref{assumption:general}, conditional on the score magnitudes and the nonnull signs, the signs of nonzero null scores are independent and their positive-sign probabilities may differ but are bounded by $c$.
The argument in Supplementary B of \citep{freestone2024semi} transfers FDR and FDX guarantees from the conditionally i.i.d.\ case, where these probabilities all equal $c$, to Assumption~\ref{assumption:general}.
Lemma~\ref{lemma:error_control_transfer} below extends this argument to a general error criterion satisfying the stated conditions, for any fixed $c\in(0,1)$.
Combined with the conditional independence and sign-probability formula in Lemma~\ref{lemma:rgkf_sign_independence}, this provides a common argument for Theorems~\ref{theorem:fdr_control} and~\ref{theorem:fdx_fwer_control} where we apply the transfer lemma with $c=1/2$ and the error criteria corresponding to FDR, FDX, and $k$-FWER.

\begin{Lemma}[Transfer of error control]
    \label{lemma:error_control_transfer}
    Fix $c\in(0,1)$. Let $\mathcal R(\mathbf u,\mathbf L)$ be a selection rule on $m$ hypotheses, with fixed scores $\mathbf u$ and labels $\mathbf L\in\{-1,+1\}^m$. Suppose that the rule selects only positive labels and does not depend on the unknown null set.
    For a null set $\mathcal N_0\subseteq[m]$, define the discovery set $R\equiv\mathcal R(\mathbf u,\mathbf L)$ and its selected-null subset $D\equiv R\cap\mathcal N_0$. Let $\ell(R,D)$ be an error metric depending only on these two sets.
    Suppose that, for every fixed score vector, every null set $\mathcal N_0\subseteq[m]$, and every fixed assignment of labels outside $\mathcal N_0$,
    \begin{equation}
        \mathbb E\!\left[\ell(R,D)\right]\leq\eta
        \label{eq:fair_sign_error_bound}
    \end{equation}
    whenever the labels on $\mathcal N_0$ are independent with $\mathbb P(L_j=1)=c$ for every $j\in\mathcal N_0$.
    Then the same bound holds when these labels are independent with $q_j\equiv\mathbb P(L_j=1)\leq c$, $j\in\mathcal N_0$.
\end{Lemma}

\begin{proof}
    Fix the scores, the null set $\mathcal N_0$, and the labels outside $\mathcal N_0$.
    Generate mutually independent variables $M_j$ and $C_j$, $j\in\mathcal N_0$, with
    \[
        M_j\sim\operatorname{Bernoulli}(q_j/c),\qquad
        C_j\sim\operatorname{Bernoulli}(c),\qquad
        L_j=2M_jC_j-1.
    \]
    Since $\mathbb P(L_j=1)=(q_j/c)c=q_j$ and the labels are independent, this gives the required joint distribution of the null labels.
    Condition on $\mathbf M$ and let $\mathcal I\equiv\{j\in\mathcal N_0:M_j=1\}$. The labels on $\mathcal I$ are independent signs with positive-sign probability $c$, whereas those on $\mathcal N_0\setminus\mathcal I$ are fixed at $-1$.
    Apply~\eqref{eq:fair_sign_error_bound} with null set $\mathcal I$, treating the hypotheses in $\mathcal N_0\setminus\mathcal I$ as false nulls, along with the false nulls outside $\mathcal N_0$. This changes neither the inputs nor the output of the selection rule. Since the rule selects only positive labels, $R\cap(\mathcal N_0\setminus\mathcal I)=\varnothing$, and therefore
    \[
        D=R\cap\mathcal N_0=R\cap\mathcal I.
    \]
    Thus the loss is unchanged, and its conditional expectation given $\mathbf M$ is at most $\eta$. Averaging over $\mathbf M$ proves the result.
\end{proof}

In particular, the following losses yield FDR, FDX, and $k$-FWER, respectively:
\begin{align}
    \ell_{\mathrm{FDR}}(R,D)&=\frac{|D|}{|R|\vee1},\notag\\
    \ell_{\mathrm{FDX},\alpha}(R,D)&=1_{\{|D|/(|R|\vee1)>\alpha\}},\notag\\
    \ell_{k\text{-}\mathrm{FWER}}(R,D)&=1_{\{|D|\geq k\}}.
    \label{eq:error_control_losses}
\end{align}
Lemma~\ref{lemma:error_control_transfer} also applies conditionally on information fixing the scores, the tested set, and the false null labels, provided the stated independence and probabilities hold conditionally as well.

\subsection{Proof of Theorem~\ref{theorem:fdr_control}}
\label{appendix:proof_theorem_fdr_control}

\begin{proof}[Proof of Theorem~\ref{theorem:fdr_control}]
    \label{proof:theorem_fdr_control}
    Condition on $\mathcal V$ in~\eqref{eq:rgkf_conditioning_field}. By Lemma~\ref{lemma:barber}, the set $\mathcal J_d=\{j\in[p]:S_j=d\}$ and the score magnitudes $|\mathbf W|$ are fixed. Every score in $\mathcal J_d$ is nonzero. By Lemma~\ref{lemma:rgkf_sign_independence}, their signs are mutually independent, and for every $j\in\mathcal H_0^d\cap\mathcal J_d$,
    \begin{equation}
        q_j\equiv\mathbb P(W_j>0\mid\mathcal V)
        =\frac{\exp(d\beta_jr_j/\sigma^2)}{1+\exp(d\beta_jr_j/\sigma^2)}
        \leq\frac12,
        \label{eq:rgkf_conservative_signs}
    \end{equation}
    because $d\beta_j\leq0$. 
    In particular, \eqref{eq:rgkf_conservative_signs} remains valid after further conditioning on the signs indexed by $\mathcal J_d\setminus\mathcal H_0^d$.

    Under independent fair null signs, SSS+ with $c=1/2$ controls the FDR for every fixed score ordering and false null label configuration by Theorem~3 of \citep{barber:controlling}. 
    Since SSS+ selects only positive labels, Lemma~\ref{lemma:error_control_transfer}, with $c=1/2$, loss $\ell_{\mathrm{FDR}}$, and $\eta=\alpha$, applies to its inputs on $\mathcal J_d$ with null set $\mathcal H_0^d\cap\mathcal J_d$. Averaging over the nonnull signs gives
    \[
        \mathbb E\!\left[
            \frac{|\mathcal{\hat S}_{\mathrm{RGKF}}\cap\mathcal H_0^d|}
            {|\mathcal{\hat S}_{\mathrm{RGKF}}|\vee1}
            \mathrel{\Big|}\mathcal V
        \right]\leq\alpha.
    \]
    Finally, averaging over $\mathcal V$ proves prespecified directional FDR control as defined in~\eqref{eq:fdr_dir_d}.
\end{proof}

\subsection{Proof of Theorem~\ref{theorem:fdx_fwer_control}}

\begin{proof}[Proof of Theorem~\ref{theorem:fdx_fwer_control}]
    Use the same conditioning as in the proof of Theorem~\ref{theorem:fdr_control}: fix $\mathcal V$ and the signs indexed by $\mathcal J_d\setminus\mathcal H_0^d$. The remaining signs are independent and satisfy~\eqref{eq:rgkf_conservative_signs}.
    Both FDP-SD and JS select only positive labels, and their guarantees under independent fair null signs hold for arbitrary fixed scores and nonnull labels.

    For FDP-SD, Theorem~1 of \citep{luo:competition} gives the bound~\eqref{eq:fair_sign_error_bound} with loss $\ell_{\mathrm{FDX},\alpha}$ and $\eta=\gamma$. Applying Lemma~\ref{lemma:error_control_transfer} with $c=1/2$ and then averaging over the conditioning information yields
    \[
        \mathbb P\!\left(
            \frac{|\mathcal{\hat S}_{\mathrm{FDX-RGKF}}\cap\mathcal H_0^d|}
            {|\mathcal{\hat S}_{\mathrm{FDX-RGKF}}|\vee1}>\alpha
        \right)\leq\gamma.
    \]
    For JS, Theorem~3.1 of \citep{janson:familywise} gives the bound~\eqref{eq:fair_sign_error_bound} with loss $\ell_{k\text{-}\mathrm{FWER}}$ and $\eta=\alpha$. The same transfer lemma and averaging give
    \[
        \mathbb P\!\left(
            |\mathcal{\hat S}_{\text{k-FWER-RGKF}}\cap\mathcal H_0^d|\geq k
        \right)\leq\alpha.
    \]
    This proves the two prespecified directional guarantees.
\end{proof}

\subsection{Closed form solution for equi-correlated response-guided knockoffs}

The equi-correlated response-guided knockoffs are constructed by solving the following optimisation problem:
\begin{align}\label{d31}
    s^* = \max \left \{ s \geq 0 \mathrel{\Big|} \mathbf G(s) \succeq 0 \right \}, \quad \text{where}
\quad
\mathbf G(s)
\equiv
\begin{bmatrix}
\boldsymbol{\Sigma}^{(e)} & \boldsymbol{\Sigma}^{(e)}- s \mathbf I(\hat{\mathcal S}_0)\\
\boldsymbol{\Sigma}^{(e)}- s \mathbf I(\hat{\mathcal S}_0) & \boldsymbol{\Sigma}^{(e)}
\end{bmatrix},
\end{align}
where $\boldsymbol{\Sigma}^{(e)}$ is defined in Condition \ref{ass:correct} and $\mathbf I(\hat{\mathcal S}_0)$ denotes the diagonal matrix with ones on the diagonal entries corresponding to $\hat{\mathcal S}_0$ and zeros elsewhere. Denote $\mathbf A_{BC}$ as the submatrix of matrix $\mathbf A$ with rows indexed by $B$ and columns indexed by $C$.

\begin{Theorem}
    The closed-form solution to \eqref{d31} is given by
    \begin{align*}
        s^* = 2 \lambda_{min} \left(\boldsymbol{\Sigma}^{(e)}_{\hat{\mathcal S}_0 \hat{\mathcal S}_0} - \boldsymbol{\Sigma}^{(e)}_{\hat{\mathcal S}_0 \hat{\mathcal S}_0^c} \left( \boldsymbol{\Sigma}^{(e)}_{\hat{\mathcal S}_0^c\hat{\mathcal S}_0^c} \right)^{-1} \boldsymbol{\Sigma}^{(e)}_{\hat{\mathcal S}_0^c \hat{\mathcal S}_0} \right) \wedge 1,
    \end{align*}
     where $\hat{\mathcal S}_0^c \equiv [p + 2] \setminus \hat{\mathcal S}_0$ corresponds to the indices of features not selected in the initial screening step plus the intercept and the perturbed response variable.
\end{Theorem}

\begin{proof}
    As previously established \citep{barber:controlling,gimenez:improving}, the following equivalence follows from Schur complements
    \begin{align}
        \label{eq:schur}
        \mathbf G(s) \succeq 0 &\iff 2 \boldsymbol{\Sigma}^{(e)} - s \mathbf I(\hat{\mathcal S}_0) \succeq 0 .
    \end{align}
    Assume without loss of generality that $\hat{\mathcal S}_0$ corresponds to the first $\lvert \hat{\mathcal S}_0 \rvert$ features in $\mathbf X^{(e)}$. 
    Then we can write $\boldsymbol{\Sigma}^{(e)}$ in the following block matrix form
    \begin{align}\label{structure}
        \boldsymbol{\Sigma}^{(e)} = \begin{bmatrix}
        \boldsymbol{\Sigma}^{(e)}_{\hat{\mathcal S}_0 \hat{\mathcal S}_0} & \boldsymbol{\Sigma}^{(e)}_{\hat{\mathcal S}_0 \hat{\mathcal S}_0^c}\\
        \boldsymbol{\Sigma}^{(e)}_{\hat{\mathcal S}_0^c \hat{\mathcal S}_0} & \boldsymbol{\Sigma}^{(e)}_{\hat{\mathcal S}_0^c \hat{\mathcal S}_0^c}\\
        \end{bmatrix}.
    \end{align}
    Since $\boldsymbol{\Sigma}^{(e)}_{\hat{\mathcal S}_0^c \hat{\mathcal S}_0^c} \succ 0$ as a submatrix of $\boldsymbol{\Sigma}^{(e)} \succ 0$, by Schur complements, we have the following
    \begin{align}\label{feas}
       2\boldsymbol{\Sigma}^{(e)} - s \mathbf I(\hat{\mathcal S}_0) \succeq 0 &\iff \begin{bmatrix}
        2\boldsymbol{\Sigma}^{(e)}_{\hat{\mathcal S}_0 \hat{\mathcal S}_0} - s \mathbf I & 2 \boldsymbol{\Sigma}^{(e)}_{\hat{\mathcal S}_0 \hat{\mathcal S}_0^c}\\
        2 \boldsymbol{\Sigma}^{(e)}_{\hat{\mathcal S}_0^c \hat{\mathcal S}_0} & 2 \boldsymbol{\Sigma}^{(e)}_{\hat{\mathcal S}_0^c \hat{\mathcal S}_0^c}
        \end{bmatrix}
        \succeq 0 \nonumber\\
        &\iff 2\boldsymbol{\Sigma}^{(e)}_{\hat{\mathcal S}_0 \hat{\mathcal S}_0} - s\mathbf I - 2\boldsymbol{\Sigma}^{(e)}_{\hat{\mathcal S}_0 \hat{\mathcal S}_0^c} \left( \boldsymbol{\Sigma}^{(e)}_{\hat{\mathcal S}_0^c \hat{\mathcal S}_0^c}\right)^{-1} \boldsymbol{\Sigma}^{(e)}_{\hat{\mathcal S}_0^c \hat{\mathcal S}_0} \succeq 0 \nonumber\\
        &\iff 2\left(\boldsymbol{\Sigma}^{(e)}_{\hat{\mathcal S}_0 \hat{\mathcal S}_0} - \boldsymbol{\Sigma}^{(e)}_{\hat{\mathcal S}_0 \hat{\mathcal S}_0^c} \left( \boldsymbol{\Sigma}^{(e)}_{\hat{\mathcal S}_0^c \hat{\mathcal S}_0^c}\right)^{-1} \boldsymbol{\Sigma}^{(e)}_{\hat{\mathcal S}_0^c \hat{\mathcal S}_0} \right) - s\mathbf I \succeq 0,
    \end{align}
    where $\mathbf I \in \mathbb R^{\lvert \hat{\mathcal S}_0 \rvert \times \lvert \hat{\mathcal S}_0 \rvert}$ denotes the identity matrix.
    Hence, the optimisation problem can be rephrased as
     \begin{align}
            s^* = \max \left \{ s \geq 0 \mathrel{\Big|} 2\left(\boldsymbol{\Sigma}^{(e)}_{\hat{\mathcal S}_0 \hat{\mathcal S}_0} - \boldsymbol{\Sigma}^{(e)}_{\hat{\mathcal S}_0 \hat{\mathcal S}_0^c} \left( \boldsymbol{\Sigma}^{(e)}_{\hat{\mathcal S}_0^c \hat{\mathcal S}_0^c}\right)^{-1} \boldsymbol{\Sigma}^{(e)}_{\hat{\mathcal S}_0^c \hat{\mathcal S}_0} \right) - s\mathbf I \succeq 0 \right \}.
        \end{align}
    From \eqref{eq:schur}, the above optimisation problem takes the form of the ordinary equi optimisation problem (\ref{d31}), where we have $\boldsymbol{\Sigma}^{(e)}_{\hat{\mathcal S}_0 \hat{\mathcal S}_0} - \boldsymbol{\Sigma}^{(e)}_{\hat{\mathcal S}_0 \hat{\mathcal S}_0^c} \left( \boldsymbol{\Sigma}^{(e)}_{\hat{\mathcal S}_0^c \hat{\mathcal S}_0^c}\right)^{-1} \boldsymbol{\Sigma}^{(e)}_{\hat{\mathcal S}_0^c \hat{\mathcal S}_0}$ in replace of $\boldsymbol{\Sigma}^{(e)}$. The closed-form solution of the ordinary equi optimisation problem is $s^* = 2\lambda_{min} \left( \boldsymbol{\Sigma}^{(e)} \right) \wedge 1$ \citep{barber:controlling} and so replacing $ \boldsymbol{\Sigma}^{(e)} $ for $\boldsymbol{\Sigma}^{(e)}_{\hat{\mathcal S}_0 \hat{\mathcal S}_0} - \boldsymbol{\Sigma}^{(e)}_{\hat{\mathcal S}_0 \hat{\mathcal S}_0^c} \left( \boldsymbol{\Sigma}^{(e)}_{\hat{\mathcal S}_0^c \hat{\mathcal S}_0^c}\right)^{-1} \boldsymbol{\Sigma}^{(e)}_{\hat{\mathcal S}_0^c \hat{\mathcal S}_0}$, we have the following solution to \eqref{d31}
    \begin{align*}
        s^* = 2 \lambda_{min} \left(\boldsymbol{\Sigma}^{(e)}_{\hat{\mathcal S}_0 \hat{\mathcal S}_0} - \boldsymbol{\Sigma}^{(e)}_{\hat{\mathcal S}_0 \hat{\mathcal S}_0^c} \left( \boldsymbol{\Sigma}^{(e)}_{\hat{\mathcal S}_0^c\hat{\mathcal S}_0^c} \right)^{-1} \boldsymbol{\Sigma}^{(e)}_{\hat{\mathcal S}_0^c \hat{\mathcal S}_0} \right) \wedge 1.
    \end{align*}
\end{proof}

\clearpage
\appendix
\renewcommand{\thesection}{C}
\renewcommand{\thealgorithm}{C\arabic{algorithm}}
\renewcommand{\thetable}{C\arabic{table}}
\renewcommand{\thefigure}{C\arabic{figure}}
\renewcommand{\theLemma}{C\arabic{Lemma}}
\renewcommand{\theCondition}{C\arabic{Condition}}
\renewcommand{\theProposition}{C\arabic{Proposition}}
\setcounter{figure}{0}
\setcounter{table}{0}
\setcounter{Proposition}{0}

\section{Further analyses and method details}

\subsection{Simulation details}
\label{appendix:sim_details}

We follow the fixed-X simulation setup of \cite{spector2022powerful}.
In each replicate, a design matrix $\mathbf X_{\text{raw}} \in \mathbb{R}^{n \times p}$ is generated with rows drawn i.i.d.\ from $\mathcal{N}(\mathbf{0}, \boldsymbol{\Sigma})$ and the response is generated as $\mathbf{y} = \mathbf X_{\text{raw}} \boldsymbol\beta + \boldsymbol\varepsilon$, $\boldsymbol\varepsilon \sim \mathcal{N}(\mathbf{0}, \mathbf{I}_n)$, using the unnormalised design matrix.
Each column of $\mathbf X_{\text{raw}}$ is then normalised to unit $\ell_2$ norm to obtain $\mathbf X$, which is used for the knockoff construction.
The $k$ nonzero entries of $\boldsymbol\beta$ are drawn independently from $\text{Unif}([\delta/2,\, \delta])$ and are assigned i.i.d.\ random signs with equal probability $1/2$ with $\delta$ set according to the six covariance structures for the features described below (\texttt{AR1}, \texttt{AR1 (Corr)}, \texttt{Block Equi}, \texttt{ER (Cov)}, \texttt{ER (Prec)}, \texttt{Equi}).
Except where noted below, the $k$ nonzero coefficients are placed uniformly at random among the $p$ features.
Each replicate is an independent draw of the data-generating model: the design $\mathbf X$, the coefficients $\boldsymbol\beta$, and the response $\mathbf y$ are sampled afresh, and for the types in which $\boldsymbol{\Sigma}$ is randomly generated, the covariance $\boldsymbol{\Sigma}$ is redrawn for each replicate as well. 
We now describe the six covariance structures $\boldsymbol{\Sigma}$ used in the simulations.

\begin{itemize}
    \item \textbf{\texttt{AR1}.} $\boldsymbol{\Sigma}$ is the covariance matrix of an AR(1) Gaussian process with heterogeneous successive correlations. Specifically, the correlation between adjacent features is $\Sigma_{j,j+1} = \rho_j$, where $\rho_1, \ldots, \rho_{p-1} \overset{\text{i.i.d.}}{\sim} \text{Beta}(3, 1)$, and the correlation between features $j$ and $k$ is $\Sigma_{jk} = \prod_{l=\min(j,k)+1}^{\max(j,k)} \rho_l$. We set $\delta = 0.45$.

    \item \textbf{\texttt{AR1 (Corr)}.} The covariance matrix is generated identically to \texttt{AR1}, but the $k$ nonzero coefficients are placed in a single consecutive block, starting at a uniformly chosen position. This mimics settings such as genetic studies in which causal variants tend to cluster along the genome. We set $\delta = 0.45$.

    \item \textbf{\texttt{Block Equi}.} $\boldsymbol{\Sigma}$ is block-equicorrelated: the $p$ features are partitioned into non-overlapping blocks of size 5, with within-block correlation $\rho = 0.5$ and zero between-block correlation.
    Formally, $\Sigma_{jk} = 0.5 \cdot \mathbf{1}_{\lfloor (j-1)/5 \rfloor = \lfloor (k-1)/5 \rfloor}$ for $j \neq k$.
    The $k$ nonzero coefficients are block-clustered: $k/5$ blocks are selected uniformly without replacement, and all five features in each selected block are assigned nonzero coefficients. We set $\delta = 0.15$.

    \item \textbf{\texttt{ER (Cov)}.} $\boldsymbol{\Sigma}$ is a sparse covariance matrix generated via an Erd\H{o}s--R\'{e}nyi procedure: each upper-triangular off-diagonal entry of a candidate matrix $\mathbf{V}$ is independently set to zero with probability $0.8$, and otherwise drawn from $\text{Unif}(0.1, 1)$ with a random $\pm 1$ sign. 
    These entries are reflected across the diagonal, and the diagonal entries are set to one. 
    We then add $c\mathbf{I}_p$, where $c = \max\{0,\,0.1-\lambda_{\min}(\mathbf{V})\}$ and $\lambda_{\min}(\mathbf{V})$ is the smallest eigenvalue of $\mathbf{V}$. This shifts every eigenvalue by $c$, so $\mathbf{V}+c\mathbf{I}_p$ is positive definite with smallest eigenvalue at least $0.1$. 
    Finally, we rescale $\mathbf{V}+c\mathbf{I}_p$ to a correlation matrix. We set $\delta = 0.5$.

    \item \textbf{\texttt{ER (Prec)}.} The \emph{precision} (inverse covariance) matrix $\boldsymbol\Omega = \boldsymbol{\Sigma}^{-1}$ is sparse and generated by the same Erd\H{o}s--R\'{e}nyi procedure as \texttt{ER (Cov)}, and $\boldsymbol{\Sigma}$ is then recovered as $\boldsymbol\Omega^{-1}$. We set $\delta = 0.25$.

    \item \textbf{\texttt{Equi}.} $\boldsymbol{\Sigma}$ is fully equicorrelated: $\Sigma_{jk} = \rho$ for all $j \neq k$, where $\rho = 0.5$. This is a special case of the block-equicorrelated design with a single block of size $p$. We set $\delta = 0.15$.
\end{itemize}

\subsection{Naive approaches simulation details}
\label{appendix:naive_sim_details}

The simulations used to evaluate the naive approaches (Figure~\ref{fig:naive_approaches}) used the following setup which is distinct from the one described in Section~\ref{appendix:sim_details}.
Throughout, we set $n = 2000$, $p = 500$, $k = 150$ nonzero coefficients, prespecified directional FDR level $\alpha = 0.1$, target sign direction $d = +1$, and ran $100$ independent replicates.

Two covariance structures are considered:
(i) the Spector--Janson \texttt{AR1} structure described in Section~\ref{appendix:sim_details};
(ii) an \texttt{Independent} design with $\boldsymbol{\Sigma} = \mathbf I_p$.
In both cases, each nonzero coefficient has magnitude drawn independently from $\text{Unif}(0.5, 1)$.
Unlike the main simulations, both the design matrix $\mathbf X$ and the coefficient vector $\boldsymbol\beta$ are drawn once and held fixed across all replicates; only the noise $\boldsymbol\varepsilon \sim \mathcal N(\mathbf 0, \mathbf I_n)$ varies.
As in Section~\ref{appendix:sim_details}, the response is generated from the unnormalised design matrix and each column is then normalised to unit $\ell_2$ norm for the knockoff construction.
Two sign configurations are considered for each covariance structure:
(i) \emph{Mostly positive}: $90\%$ positive, $10\%$ negative;
(ii) \emph{Mostly negative}: $90\%$ negative, $10\%$ positive.
This gives four scenarios in total.

For the unmodified knockoff filter, post-hoc sign screening, and pre-screening methods, equi-correlated fixed-$X$ knockoffs are constructed using the \texttt{knockoff} R package \citep{barber:controlling}.
Method (d), the signed-knockoff filter, and method (e), the response-guided knockoff filter, use the same R implementation as described in the main text with the ME knockoff construction.

Figure~\ref{fig:naive_approaches} reports the empirical prespecified directional FDR of each naive strategy alongside the signed-knockoff filter (denoted \texttt{S(ME)}, Section~\ref{section:signed-knockoff-filter}) and the response-guided knockoff filter (denoted \texttt{RG(signed, ME)}, Section~\ref{section:response-guided_kf}).
All three naive strategies substantially exceed the nominal prespecified directional FDR level $\alpha$ in at least one scenario, while \texttt{S(ME)} and \texttt{RG(signed, ME)} consistently control the prespecified directional FDR at or below the target level.

\begin{figure}[ht!]
    \centering
    \includegraphics[width=6in]{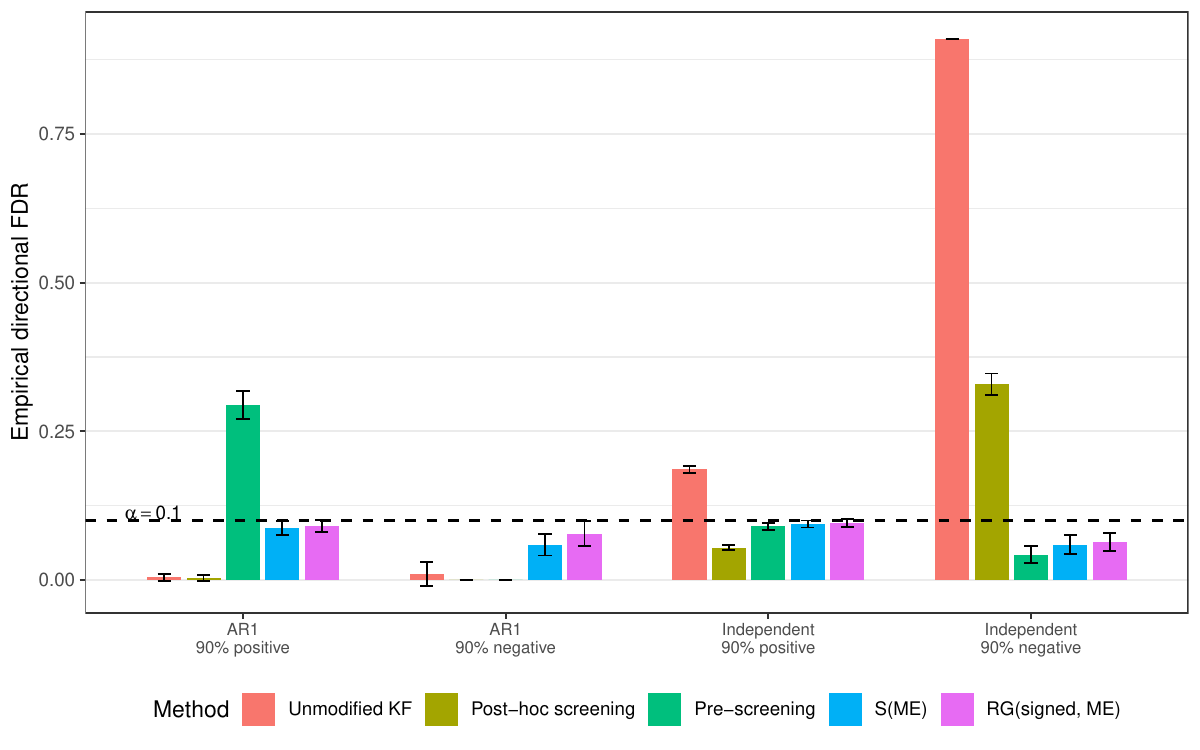}
    \caption{\textbf{Empirical prespecified directional FDR of three naive approaches compared with the signed-knockoff filter and response-guided knockoff filter.}
    Each group of bars corresponds to a different combination of covariance structure (AR1 or independent) and sign configuration ($n = 2000$, $p = 500$, $k = 150$, $\alpha = 0.1$); see Appendix~\ref{appendix:naive_sim_details} for details.
    The dashed line marks the nominal prespecified directional FDR threshold $\alpha$; FDR labels in the plot refer to prespecified directional FDR.
    All three naive strategies exceed $\alpha$ in at least one scenario, while the signed-knockoff filter (Section \ref{section:signed-knockoff-filter}) and the response-guided knockoff filter (Section \ref{section:response-guided_kf}) consistently control the prespecified directional FDR.}
    \label{fig:naive_approaches}
\end{figure}

\subsection{Per-experiment plots}
\label{appendix:power_plots}

Figures~\ref{fig:spector_simulation_power}, \ref{fig:spector_simulation_2_power}, and~\ref{fig:hiv_power} present the detailed power curves underlying the relative-power summary of Figure~\ref{fig:tpr_summary}, for Simulation~1, Simulation~2, and the HIV drug-resistance data, respectively.
In Figures~\ref{fig:spector_simulation_power} and \ref{fig:spector_simulation_2_power}, we report the average empirical power of each method taken over the number of observations $n$ and the number of nonzero coefficients $k$.

Figures~\ref{fig:spector_simulation_fdr} and~\ref{fig:spector_simulation_2_fdr} present the complementary empirical prespecified directional FDR plots for Simulations~1 and~2, respectively.
The empirical prespecified directional FDR $\widehat{\mathrm{FDR}}_{\mathrm{dir}}^{d}$ is evaluated at $(n, k) = (1500, 100)$ for Simulation~1 and $(n, k) = (700, 100)$ for Simulation~2, averaging over the 100 Monte Carlo repetitions separately for each linear-model type, method, and target level $\alpha$.

\begin{figure}[h!]
    \centering
    \begin{tabular}{l}
    \includegraphics[width=6in]{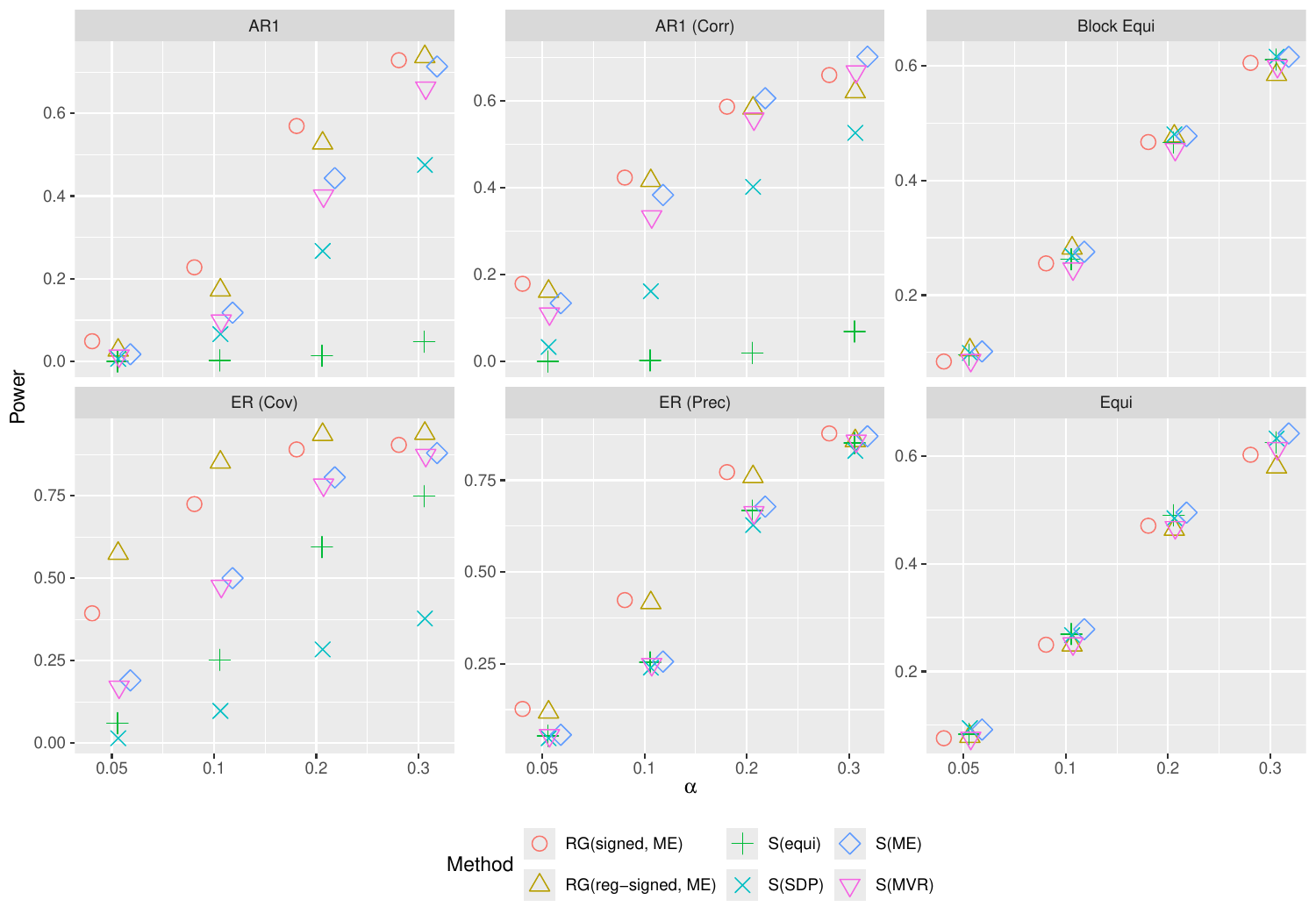}\\
    \end{tabular}
    \caption{\textbf{Comparing the power of response-guided knockoff and signed-knockoff filters on six types of linear models ($n\ge2p$).}
    Each panel corresponds to a different type of linear model as described in Appendix~\ref{appendix:sim_details}.
    The $x$-axis shows the prespecified directional FDR level $\alpha$ in $\{0.05, 0.1, 0.2, 0.3\}$ and the $y$-axis shows the empirical power defined in (\ref{eq:empirical_power_fdr}), averaged over $n$ and $k$. FDR labels in the plot refer to prespecified directional FDR. }
    \label{fig:spector_simulation_power}
\end{figure}

\begin{figure}[h!]
    \centering
    \begin{tabular}{l}
    \includegraphics[width=6in]{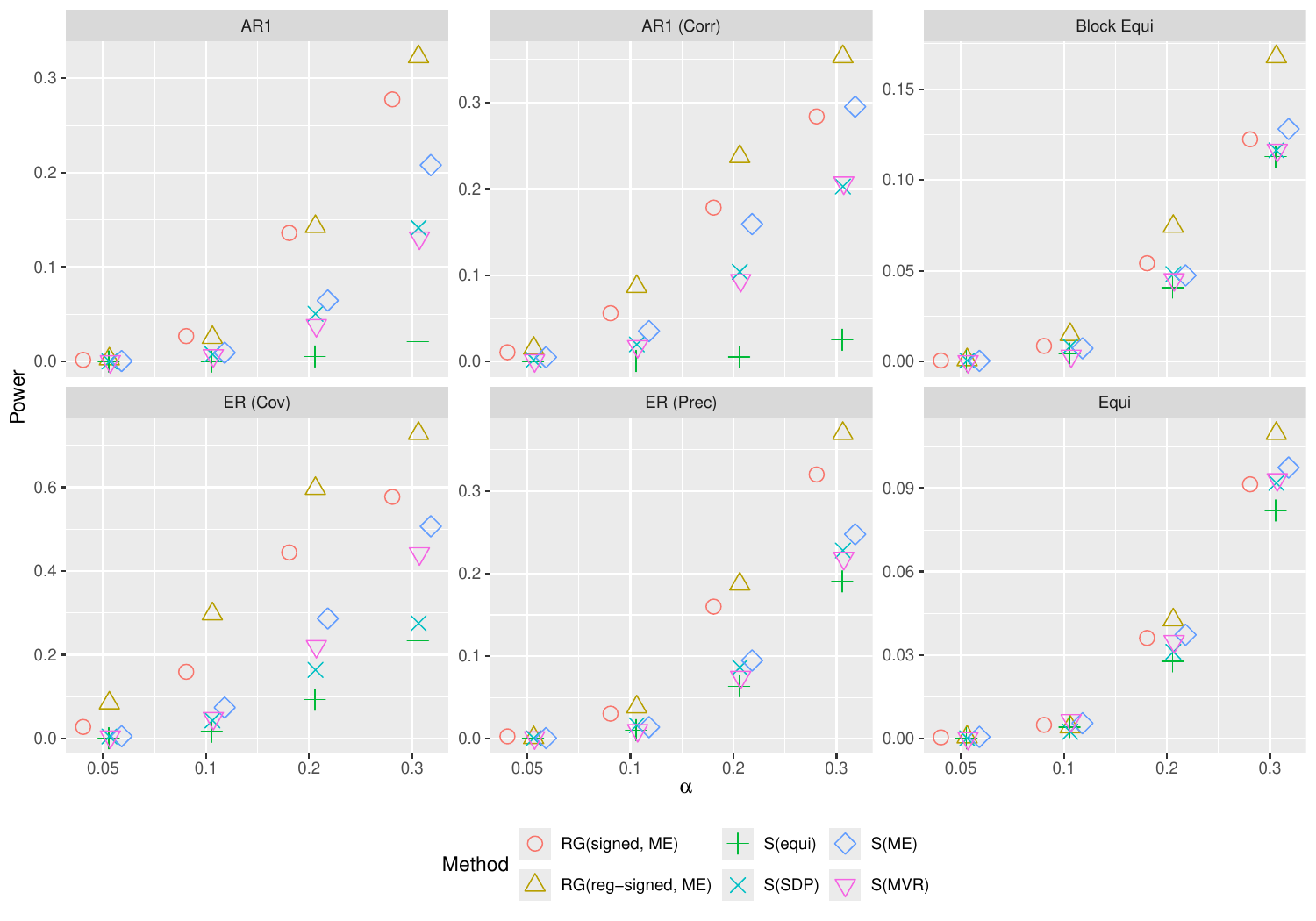}\\
    \end{tabular}
    \caption{\textbf{Comparing the power of response-guided knockoff and signed-knockoff filters on six types of linear models when $p + 2 < n < 2p$.}
    Each panel corresponds to a different type of linear model as described in Appendix~\ref{appendix:sim_details}.
    The $x$-axis shows the prespecified directional FDR level $\alpha$ in $\{0.05, 0.1, 0.2, 0.3\}$ and the $y$-axis shows the empirical power defined in (\ref{eq:empirical_power_fdr}), averaged over $n$ and $k$. FDR labels in the plot refer to prespecified directional FDR.}
    \label{fig:spector_simulation_2_power}
\end{figure}

\begin{figure}[h!]
    \centering
    \begin{tabular}{l}
    \includegraphics[width=6in]{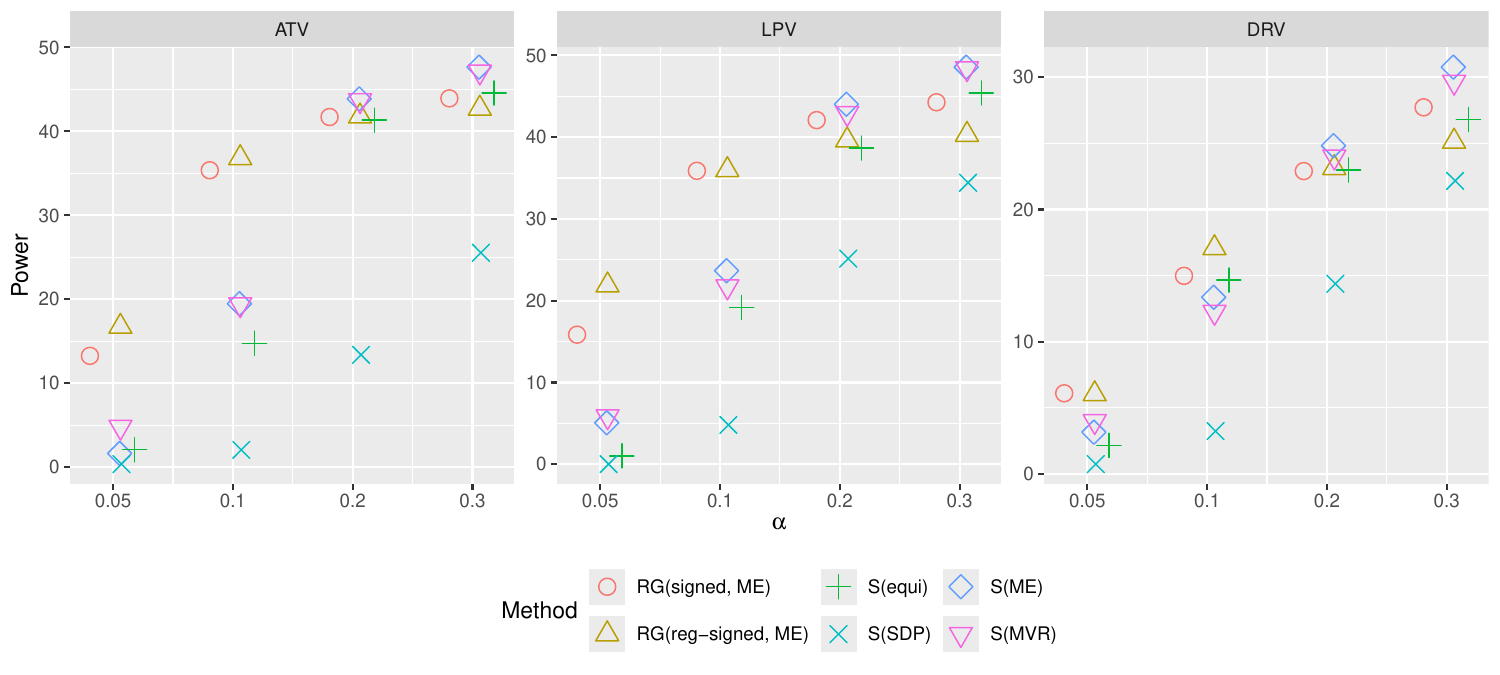}\\
    \end{tabular}
    \caption{\textbf{Comparing the discovery count of response-guided knockoff and signed-knockoff filters on HIV drug resistance data.}
    Each panel corresponds to a different PI drug.
    The $x$-axis shows the prespecified directional FDR level $\alpha$ in $\{0.05, 0.1, 0.2, 0.3\}$ and the $y$-axis shows the estimated number of true discoveries. FDR labels in the plot refer to prespecified directional FDR.}
    \label{fig:hiv_power}
\end{figure}

\begin{figure}[h!]
    \centering
    \begin{subfigure}{\textwidth}
        \centering
        \includegraphics[width=5.3in]{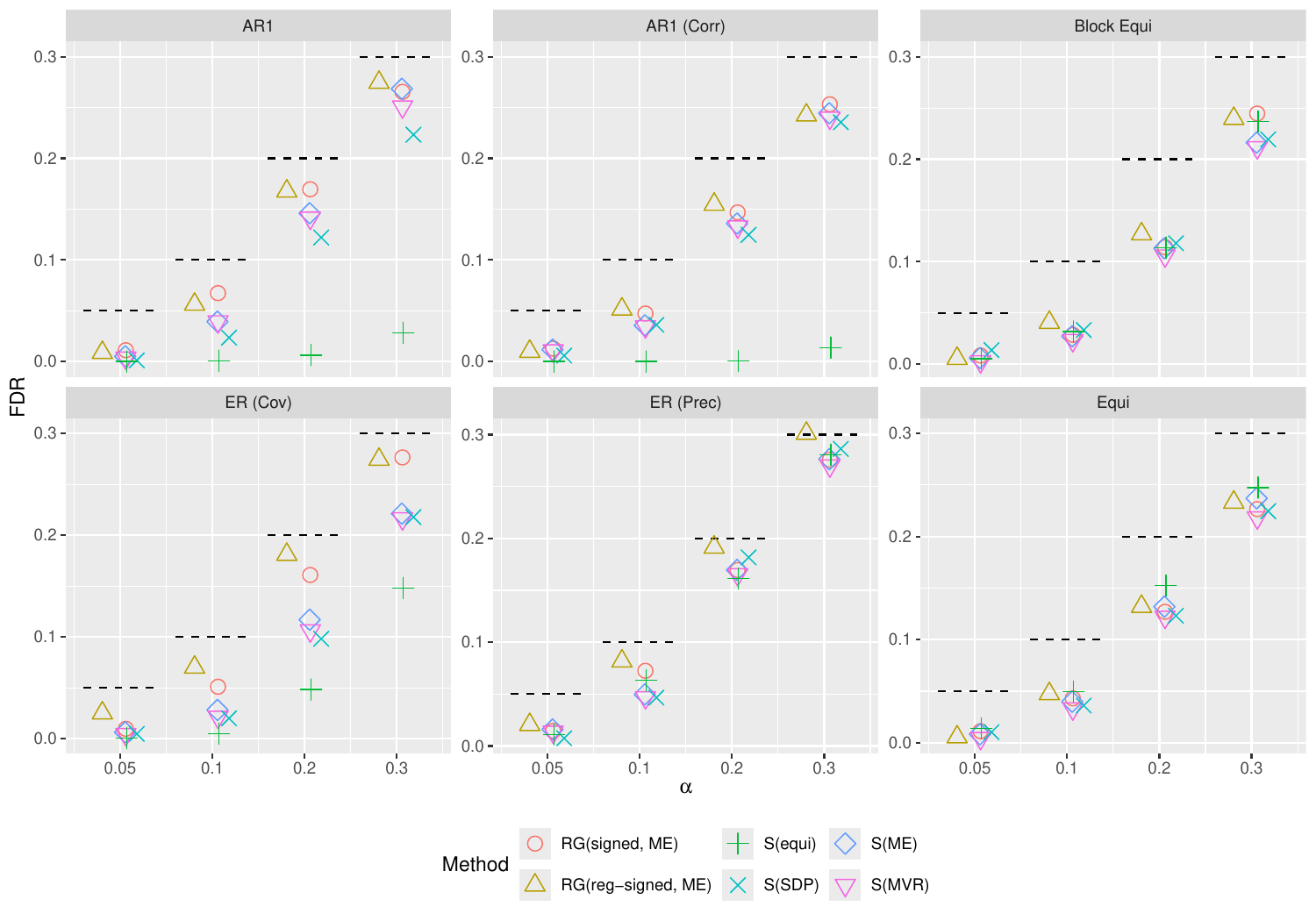}
        \caption{Simulation~1 ($n = 1500$, $k = 100$).}
        \label{fig:spector_simulation_fdr}
    \end{subfigure}

    \vspace{1em}

    \begin{subfigure}{\textwidth}
        \centering
        \includegraphics[width=5.3in]{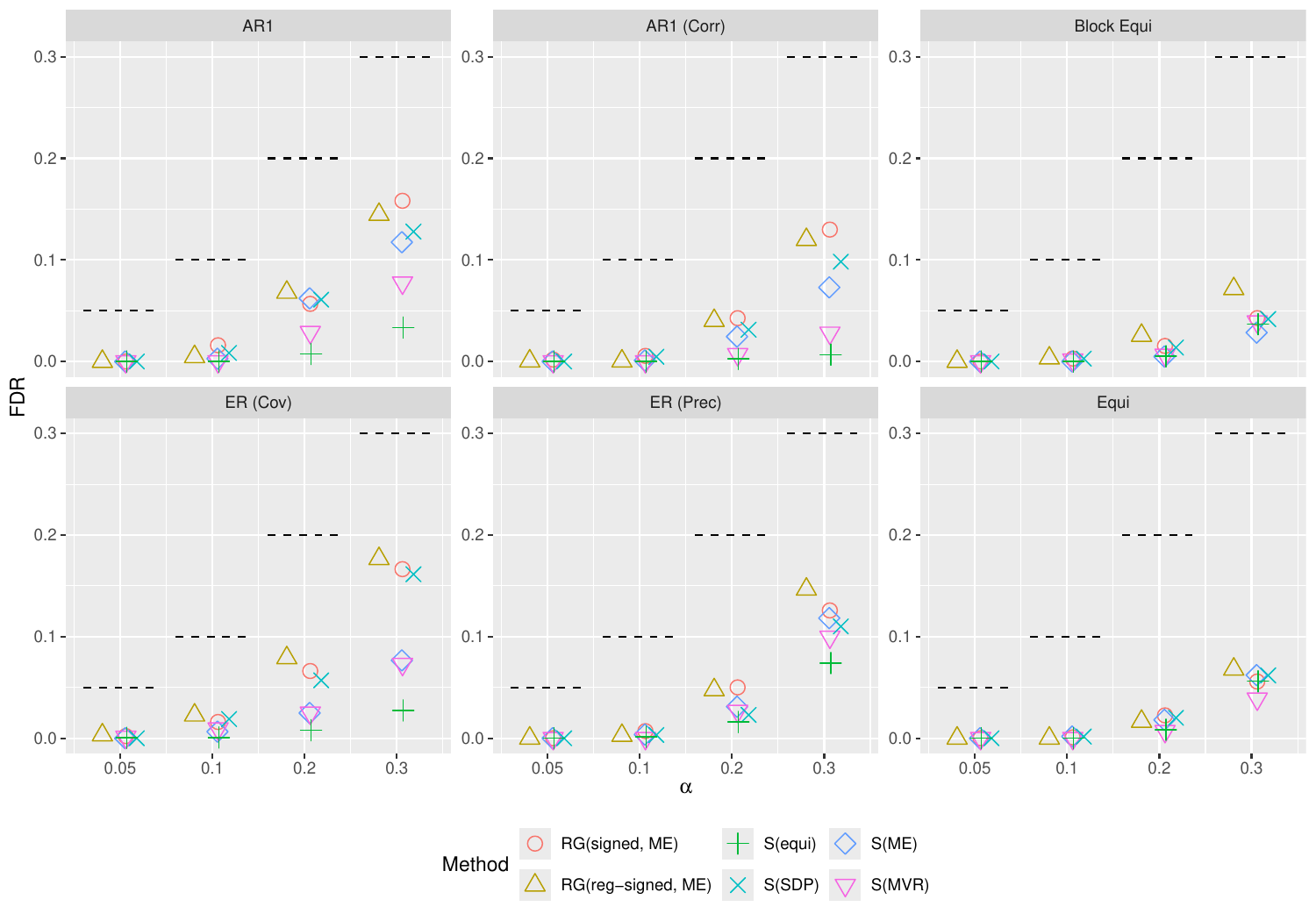}
        \caption{Simulation~2 ($n = 700$, $k = 100$).}
        \label{fig:spector_simulation_2_fdr}
    \end{subfigure}

    \caption{\textbf{Empirical prespecified directional FDR of response-guided knockoff and signed-knockoff filters on six types of linear models.}
    Each panel corresponds to a different type of linear model as described in Appendix~\ref{appendix:sim_details}.
    The $x$-axis shows the prespecified directional FDR level $\alpha$ in $\{0.05, 0.1, 0.2, 0.3\}$ and the $y$-axis shows the empirical prespecified directional FDR defined in (\ref{eq:empirical_power_fdr}), evaluated at the fixed $(n, k)$ setting indicated in each subcaption.
    The dashed line marks the nominal prespecified directional FDR level $\alpha$; FDR labels in the plots refer to prespecified directional FDR.}
    \label{fig:fdr_plots}
\end{figure}

\clearpage
\subsection{HIV drug resistance data analysis details}
\label{appendix:hiv_details}

\begin{table}[htbp]
    \centering
    \small
    \setlength{\tabcolsep}{4pt}
    \renewcommand{\arraystretch}{1.15}
    \begin{tabularx}{\linewidth}{>{\raggedright\arraybackslash}p{2.7cm} >{\raggedright\arraybackslash}X >{\raggedright\arraybackslash}X >{\raggedright\arraybackslash}X}
        \toprule
        Drug & ATV & LPV & DRV \\
        \midrule
        Approximate Ground truth &
        47A, 54A, 73A, 82A, 84A, 73C, 82C, 84C, 83D, 88D, 58E, 10F, 24F, 33F, 82F, 88G, 10I, 11I, 20I, 23I, 24I, 32I, 33I, 46I, 71I, 82I, 11L, 46L, 50L, 53L, 54L, 71L, 82L, 20M, 24M, 48M, 54M, 82M, 90M, 30N, 74P, 10R, 20R, 54S, 73S, 74S, 82S, 88S, 20T, 43T, 54T, 71T, 73T, 82T, 89T, 10V, 20V, 33V, 46V, 47V, 48V, 50V, 54V, 71V, 76V, 84V, 85V, 89V, 10Y, 53Y
        & 
        47A, 54A, 73A, 82A, 84A, 73C, 82C, 84C, 83D, 88D, 58E, 10F, 24F, 33F, 82F, 88G, 10I, 11I, 20I, 23I, 24I, 32I, 33I, 46I, 71I, 82I, 11L, 46L, 53L, 54L, 71L, 82L, 20M, 24M, 48M, 54M, 82M, 90M, 30N, 74P, 10R, 20R, 54S, 73S, 74S, 82S, 88S, 20T, 43T, 54T, 71T, 73T, 82T, 89T, 10V, 20V, 33V, 46V, 47V, 48V, 50V, 54V, 71V, 76V, 84V, 85V, 89V, 10Y, 53Y
        & 
        47A, 54A, 73A, 73C, 82C, 83D, 88D, 58E, 10F, 24F, 33F, 82F, 88G, 10I, 11I, 20I, 23I, 24I, 32I, 33I, 46I, 71I, 82I, 11L, 46L, 53L, 54L, 71L, 82L, 20M, 24M, 48M, 54M, 82M, 90M, 30N, 74P, 10R, 20R, 54S, 73S, 74S, 20T, 43T, 54T, 71T, 73T, 89T, 10V, 20V, 33V, 46V, 47V, 48V, 50V, 54V, 71V, 76V, 84V, 85V, 89V, 10Y, 53Y
        \\
        Susceptible &
        --- & 50L & 50L, 82A, 82T, 82S, 88S \\
        \bottomrule
    \end{tabularx}
    \caption{\textbf{List of annotated mutations for each PI drug.}}
    \label{table:hiv_mutations}
\end{table}

Drug resistance data was downloaded from \url{https://hivdb.stanford.edu/_wrapper/download/GenoPhenoDatasets/PI_DataSet.txt}.
The data was then pre-processed according to the tutorial at \url{https://web.stanford.edu/group/candes/knockoffs/software/knockoffs/tutorial-4-r.html} which exactly carries out the same pre-processing in \cite{barber:controlling}.
A design matrix and response vector was obtained for each of the three main protease inhibitors: atazanavir (ATV), lopinavir (LPV) and darunavir (DRV).
An approximate ground truth of mutations for each drug was obtained by first downloading the list of mutations from \url{https://hivdb.stanford.edu/dr-summary/comments/PI/}.
For each drug, we defined an approximate ground truth mutation to be any mutation in the list unless it specifically indicated that the mutation was associated with drug susceptibility instead of resistance. 
Table~\ref{table:hiv_mutations} summarises our list of annotated mutations for each drug in the design matrix after pre-processing.
The number of samples $n$, number of mutations $p$ and the size of the ground truth set for each drug are summarized in Table~\ref{table:hiv_dimensions}.

\begin{table}[h]
    \centering
    \begin{tabular}{lccc}
        \hline
        Drug & ATV & LPV & DRV \\
        \hline
        $n$ (samples) & 1486 & 1788 & 993 \\
        $p$ (mutations) & 325 & 347 & 267 \\
        Ground truth size & 70 & 69 & 63 \\
        \hline
    \end{tabular}
    \caption{\textbf{Dimensions of the HIV drug resistance data for each PI drug.}}
    \label{table:hiv_dimensions}
\end{table}

\clearpage

\end{document}